\documentclass{article}

\input{style.sty}

\title{Well-conditioned iterative methods for large open quantum systems}

\author[1]{Gaspard Beugnot\thanks{gaspard.beugnot@alice-bob.com}}
\author[1]{Paul Gregory\thanks{paul.gregory.w@gmail.com}}
\author[2]{Rémi Robin\thanks{remi.robin@minesparis.psl.eu}}
\author[2]{Antoine Tilloy\thanks{antoine.tilloy@minesparis.psl.eu}}

\affil[1]{Alice $\&$ Bob, Paris, France}
\affil[2]{Laboratoire de Physique de l’\'Ecole Normale Supérieure, Mines Paris, Inria, CNRS, ENS-PSL, Sorbonne Université, PSL Research University, Paris, France}

\date{}

\begin{document}

\maketitle

\begin{abstract}
    Markovian open quantum systems are well modeled by the Lindblad Master Equation (ME) $\frac{\upd}{\upd t} \rho_t =  \mathcal{L} \rho_t$, where $\mathcal{L}$ is a linear (super-)operator and $\rho_t$ is the system state, a positive matrix. When designing or characterizing a quantum system, one is usually interested in the steady state $\rho_\infty$ (such that $\mathcal{L} \rho_\infty = 0$), the first few excited states, and trajectories $t\mapsto \rho_t$. In finite dimension, $\rho_t$ is an $n\times n$ matrix, $\mathcal{L}$ thus typically costs $n^4$ to store explicitly as a dense matrix, and $O(n^6)$ to diagonalize or invert exactly, making standard linear algebraic techniques expensive for large systems. However, $\mathcal{L}$ usually costs only $O(n^3)$ to apply. This makes iterative methods appealing, but they do not work without a good preconditioner. In this article, our main observation is that a part of the Lindblad equation, corresponding to the so-called no-jump evolution $\mathcal{S}$, can be inverted efficiently. Using this inverse map, we introduce an auxiliary completely positive trace-preserving (CPTP) map $\Phi$ whose fixed point is directly related to $\rho_\infty$, all the other eigenvalues having smaller magnitude. The map $\Phi$ is thus well suited to iterative methods, and $\rho_\infty$ can be found in a few Arnoldi iterations. Using the same inverse map $\mathcal{S}^{-1}$ as preconditioner, we compute the low-lying spectrum efficiently via shift-invert Arnoldi, and, as a proof of concept, build an implicit time integrator that is competitive on stiff systems in the low-precision regime. For the steady-state and low excited states problems, our methods scale like $O(n^3)$ per iteration and offer state-of-the-art performance on CPU and GPU.
\end{abstract}

\section{Introduction}

\subsection{Motivations and scope of the paper}
The Gorini--Kossakowski--Sudarshan--Lindblad (GKSL) equation~\cite{gorini1976,Lindblad1976} is the most general dynamics that can be followed by a Markovian open quantum system. For a density matrix $\rho$ acting on a finite dimensional complex Hilbert space $\mathscr{H}= \mathbb{C}^n$, the GKSL equation reads, in its standard form,
\begin{align}\label{eq:lindblad_def}
    \frac{\upd \rho_t}{\upd t} = \mathcal{L}(\rho_t) := -i[H,\rho_t] + \sum_{j=1}^d \left( L_j\rho_t L_j^\dagger - \frac{1}{2}\{L_j^\dagger L_j,\rho_t\}\right)
\end{align}
where $H = H^\dagger$ is the system Hamiltonian, and the $L_j$ are arbitrary matrices which parameterize the so-called dissipators $\mathcal{D}[L_j](\rho) :=  L_j\rho L_j^\dagger - \frac{1}{2}\{L_j^\dagger L_j,\rho\}$. The linear operator $\mathcal{L}$ generating the dynamics is called the Lindbladian, or sometimes Liouvillian. One can span all possible such dynamical generators $\mathcal{L}$ with $d\leq n^2 - 1$ dissipators.

Besides this mathematical generality in the Markovian case, the GKSL equation often provides an excellent approximate description for physical quantum systems coupled to realistic environments, even if they are not strictly Markovian \cite{breuer2002theory}. As a result, the GKSL equation is central for the quantitative understanding of open quantum systems, and is used routinely to model a wide range of quantum experiments, whenever decoherence and dissipation matter.

In practice, there are three main questions one is interested in answering for Markovian open quantum systems. The first is to know the steady state $\rho_\infty$ of the dynamics, that is a positive semi-definite matrix with trace $1$ such that $\mathcal{L}(\rho_\infty) = 0$; such a state always exists in finite dimension. Assuming that it is unique, a fairly common case, this is the state reached by the system at long times, after transients have decayed. This is sometimes how quantum states of interest are prepared: one simply designs a dissipative dynamics $\mathcal{L}$ of which they are the stationary points. Further, observables in the steady state are easy to evaluate experimentally,
thus used extensively to calibrate the physical parameters of a system \cite{ferrari2023steadystatequantumchaosopen,Dahan2022,KesslerPRA12,CarmichaelPRX15,MingantPRA18_Spectral,Hartmann2006}. Hence, efficient steady-state computation is critical. Interestingly, steady states of the GKSL equation also appear in the context of tensor network states. In particular, they enter into the expression of expectation values of observables of continuous matrix product states (CMPS) in the thermodynamic limit \cite{verstraete2010,haegeman2013}. Having an efficient method to compute steady states would thus allow working with far larger CMPS.

The second problem is to compute the spectral gap of the Lindbladian, or more generally its first few eigenvalues, because they determine the speed of convergence to the steady state.

The last problem is to compute the real-time dynamics, \ie to provide a good approximation of the solution of \cref{eq:lindblad_def} with given initial condition $\rho_0$.

Some example applications presented in this article involve an underlying infinite-dimensional Hilbert space, such as bosonic systems; in these cases, we truncate the Hilbert space to a finite dimension and do not consider the resulting truncation errors (see \cref{sec:cat} for details). Ultimately, we thus assume the Hilbert space dimension $n$ to be finite, but possibly quite large \ie potentially $n\geq 10^3$.

The GKSL equation is a simple linear equation, and thus it would seem the three problems we just mentioned should be trivial to solve using standard linear algebraic techniques (exact diagonalization, and exponentiation of $\mathcal{L}$). Indeed, one can answer these questions exactly by seeing $\mathcal{L}$ as a linear operator on a vector space of dimension $n^2$. With the standard implementation of matrix multiplication, that is, excluding fast matrix multiplication à la Strassen or Coppersmith-Winograd~\cite{blaser2013fastmatmul}, this implies a solution at cost $O((n^2)^3)=O(n^6)$ for all three problems. Unfortunately, this cost becomes prohibitive for $n \gtrsim 100$, which is common for properly discretized systems of a few bosons.

Fortunately, the Lindblad operator $\mathcal{L}$ is cheaper to apply than one may expect, because it is often sparse in two ways. First, the Hamiltonian $H$ and generators $L_j$ can be sparse matrices, for example if they can be expressed as simple combinations of discretized ladder operators. This sparsity is common in practice for bosonic systems, but not systematic, and thus we will not require it. The second source of sparsity is that the number $d$ of generators $L_j$ is usually $O(1)$, far from saturating the upper bound $n^2 - 1$. This, on the other hand, is the typical situation for the open quantum systems one encounters in practice. In the latter case, applying $\mathcal{L}$ to a state is done at the cost of matrix multiplication, that is $n^3$ (instead of the naive $n^4$).

This well-known observation opens the way to using efficient iterative techniques based on Krylov subspaces to solve the three problems mentioned above, at cost $O(n^3)$. Unfortunately, in many situations we suffer from the fact that $\LL$ is ill-conditioned, which is expected when solving a discretization of an initially unbounded Lindbladian. Thus, using iterative methods directly on $\LL$ gives disappointing results. The main idea of this paper is to leverage the computationally cheap inversion of the no-jump part to provide efficient iterative techniques. %

This paper is organized as follows. In \cref{sec:fixed_point} we review existing approaches for computing the steady state and introduce a new method based on the fixed point of a Completely Positive, Trace-Preserving (CPTP) map. In \cref{sec:preconditioner} we build a preconditioner from the solution of a Lyapunov equation and relate it to the resolvent of the Lindbladian. In \cref{sec:applications} we apply this preconditioner to three tasks: computing the steady state as the solution of a linear system, computing the low-lying spectrum of $\LL$ through a shift-invert transformation, and building an implicit method for time evolution. Each application is presented together with the relevant state of the art and with numerical benchmarks on representative examples. The appendix provides additional details: \cref{sec:num_implementation} describes classical methods for solving the Lyapunov equation, \cref{sec:details_num_benchmark} gives the system parameters and implementation details used in the benchmarks, and \cref{app:properties_phi} provides a proof of the fact that the map defined in \cref{sec:fixed_point} is indeed a trace preserving map.

\subsection{Notations and preliminaries}
We denote superoperators in calligraphic font (e.g. $\mathcal{L}$) or capital Greek letters (e.g. $\Phi$) and operators on the Hilbert space by capital letters (e.g. $H$, $L_j$). On the latter, we consider the operator norm $\|A\|_\infty := \sup_{\|x\|=1} \|Ax\|$, the trace-norm $\|A\|_1 := \Tr{\sqrt{A^\dagger A}}$, and the entrywise max-norm $\|A\|_{\max} := \max_{i,j} |A_{ij}|$ (which is only used as a numerical stopping criterion in this paper). On superoperators, we consider the operator norm induced by the trace norm, \ie $\|\mathcal{L}\| := \sup_{\|\rho\|_1=1} \|\mathcal{L}(\rho)\|_1$. The adjoint of an operator (with respect to Hilbert space scalar product) $A$ is denoted $A^\dagger$, while the adjoint of a super-operator (with respect to the Hilbert--Schmidt scalar product) $\mathcal{X}$ is denoted
$\mathcal{X}^*$.

The Lindblad equation \eqref{eq:lindblad_def} can be decomposed as
\begin{align}\label{eq:lindblad_decomposition}
    \mathcal{L}(\rho) & = \SS(\rho)+ \KK(\rho),                                               \\
    \SS(\rho)         & \coloneq G\rho + \rho G^\dag, \qquad
    \KK(\rho) \coloneq \sum_{j=1}^d L_j\rho L_j^\dagger,                                      \\
    G                 & \coloneq -iH - \tfrac{1}{2}\sum_{j=1}^d L_j^\dagger L_j. \label{eq:G}
\end{align}
Here $\KK$ is called the generator of the jump part of the Lindbladian and $\SS$ the generator of the no-jump part. Note that both generate completely positive maps, i.e. $e^{t\SS}$ and $e^{t\KK}$ are Completely Positive (CP) maps for all $t\geq 0$.

In this paper, we will repeatedly need to invert $\SS$, that is to solve the so-called continuous Lyapunov equation
\begin{align}
    \label{eq:lyapunov}
    \SS(X)=GX + XG^\dag = Y,
\end{align}
for a given operator $Y$.
This is a well studied problem in control theory, see \cite{10.1093/imamci/9.4.275} for a historical perspective. Note that from \cref{eq:G}, it is clear that the spectrum of $G$ always lies in the closed left half-plane. Under the assumption that $G$ has no eigenvalue on the imaginary axis, the solution is unique and can be expressed as
\begin{align}
    \label{eq:lyapunov_solution} X =\SS^{-1}(Y)= -\int_0^\infty e^{tG}Ye^{tG^\dag}\upd t.
\end{align}
In particular, $-\SS^{-1}$ is completely positive, being an integral of maps of the form $Y \mapsto e^{tG} Y e^{tG^\dag}$.

Solving such a Lyapunov equation is known to be computationally efficient, namely $O(n^3)$ while one might naively assume that, as the inversion of a superoperator, it would require $O(n^6)$ operations. We give more details on the algorithms in \cref{sec:continuous_lyapunov}.

\begin{remark}[Non-uniqueness of the decomposition]
    \label{rk:gauge}
    The decomposition in \cref{eq:lindblad_decomposition} is not unique, and the choice matters: the split is what decides how much of $\LL$ we invert exactly through a Lyapunov solve, and how much is left to be resolved. The freedom to exploit is the shift of each jump operator by a multiple of the identity. Given $c \in \CC^d$, set
    \begin{align}\label{eq:lindblad_gauge}
        \tilde L_j & \coloneq L_j + c_j \Id, \qquad j = 1, \dots, d,                                                                  \\
        \tilde H   & \coloneq H + \frac{1}{2i} \sum_{j=1}^d \p{ c_j^* \tilde L_j - c_j \tilde L_j^\dag }, \label{eq:lindblad_gauge_2}
    \end{align}
    the latter being Hermitian, and let $\tilde G$ be defined from $(\tilde H, \tilde L_j)$ as in \cref{eq:G}. A direct computation then gives
    \begin{align}\label{eq:lindblad_gauge_split}
        \LL(\rho) = \underbrace{\tilde G \rho + \rho \tilde G^\dag}_{\tilde \SS(\rho)} + \underbrace{\sum_{j=1}^d \tilde L_j \rho \tilde L_j^\dag}_{\tilde \KK(\rho)} .
    \end{align}
    Every $c$ thus yields an admissible no-jump/jump split, \cref{eq:lindblad_decomposition} being the one with $c = 0$. Precisely, $\tilde G = G - B - \frac{1}{2}\sum_j |c_j|^2$ with $B \coloneq \sum_j c_j^* L_j$, so the shift moves $B\rho + \rho B^\dag + \sum_j |c_j|^2 \rho$ from the no-jump part to the jump part, the latter remaining a Kraus map by construction.

    Two further gauges leave the split untouched: the first is an isometric mixing of the jump operators, $\tilde L_i \coloneq \sum_j u_{ij} L_j$ for $i = 1, \dots, m$ with $m \geq d$ and $u^\dag u = I_d$, the second is a global phase $\tilde H \mapsto \tilde H - \gamma \Id$, $\gamma \in \mathbb{R}$. They are not relevant here since they leave the map $\KK$ unchanged.

    The gauge freedom in \eqref{eq:lindblad_gauge} and \eqref{eq:lindblad_gauge_2} has been leveraged before, \eg in \cite{Cao2025Sep} to optimize the unravelling of the Lindblad equation. We stick throughout to the split of \cref{eq:lindblad_decomposition}, and discuss in the conclusion the potential benefits of optimizing over $c$.
\end{remark}

\section{Computing the steady state as the fixed point of CPTP maps}
\label{sec:fixed_point}
\subsection{State of the art}
We recall that the dimension of the kernel of $\LL$ is at least one in finite dimension, and that $\ker\LL$ always contains a density matrix \cite[Chapter 6]{wolf_guided_tour_2012}. We assume in this section that this dimension is exactly one. In this case, there exists a unique steady state $\rho_\infty$, i.e. $\rho_\infty \succeq 0$ with $\LL \rho_\infty = 0$ and $\Tr \rho_\infty = 1$. Standard approaches to compute the steady state can be classified into three main categories: diagonalization of the Lindbladian, solving a linear problem, or simulating the long-time dynamics.

\paragraph{Diagonalization of the Lindbladian.}
One finds $\rho_\infty$ as the eigenvector of $\mathcal{L}$ with eigenvalue of lowest magnitude (here, $0$). This can be instantiated in two ways:
\begin{itemize}
    \item \textbf{Dense direct solvers}: one constructs explicitly the Liouvillian $\LL$ as a $n^2 \times n^2$ matrix and diagonalizes it using standard linear algebra libraries. Alternatively, one can diagonalize $\LL^* \LL$ which is positive semi-definite, and exploit this structure.
    \item \textbf{Iterative methods}: Using the Arnoldi method, one can iteratively find the states with eigenvalue of lowest magnitude applying only $\LL$ as a function. Iterative techniques are fast when one targets eigenstates corresponding to eigenvalues with \emph{largest} magnitude, and thus the performance of this approach is typically not better than the dense approach above, unless the shift-invert method is used. This amounts to solving a linear system involving the Lindbladian, which is expensive and requires a good preconditioner to be efficient (see \cref{sec:low_lying}).
\end{itemize}

\paragraph{Solving a linear problem.}
Since the identity is a left eigenvector of the Lindbladian, it is straightforward to check that $\rho_\infty$ is the unique solution of the following linear problem:
\begin{align}
    \label{eq:steady_state_linear_problem}
    \widetilde \LL_\eta \rho                & =\eta \Id,                                  \\
    \text{where } \widetilde \LL_\eta(\rho) & \coloneq \LL(\rho) + \eta \Tr(\rho) \; \Id,
\end{align}
where $\eta>0$ is an arbitrary positive constant. In the following, we sometimes fix $\eta=1$ and denote $\widetilde \LL=\widetilde \LL_1$. Classical methods to solve \cref{eq:steady_state_linear_problem} include:
\begin{itemize}
    \item \textbf{Dense direct solvers}: Constructing the Liouvillian superoperator as a dense matrix ($n^2 \times n^2$) and using LU/SVD. Scales as $O(n^6)$.
    \item \textbf{Sparse direct solvers}: If the Hamiltonian and jump operators are sparse, $\LL$ is also sparse, and so is the pinned system $\widetilde \LL_\eta$\footnote{QuTiP pins the trace condition with the slightly different term $\eta \Tr(\rho)\,\ketbra{0}{0}$, and right-hand side $\eta \ketbra{0}{0}$ in place of $\eta \Id$. Any operator $P$ with $\Tr P \neq 0$ substituted for $\Id$ in \cref{eq:steady_state_linear_problem} leaves $\rho_\infty$ as the unique solution, and the choice $\ketbra{0}{0}$ makes the system sparser still. With our own solver we found the symmetrized deflation $\eta \Tr(\rho) \Id$ to perform best.}, which can be exploited using sparse linear solvers such as sparse LU (as implemented in QuTiP). A more advanced and parallelizable alternative is to use MUltifrontal Massively Parallel sparse direct Solver (MUMPS) \cite{MUMPS:1,MUMPS:2}.
    \item \textbf{Iterative methods}: Using GMRES (or another variant) on the linear system. Unfortunately, for many systems of interest it does not work properly without a good preconditioner. QuTiP~\cite{johansson2012qutip}, for example, proposes an incomplete LU preconditioner, a general-purpose preconditioner for sparse systems. It may however not always be adapted to the specific structure of the Lindblad equation. Speedup can be obtained in specific cases, when computing multiple related steady states. The cost of computing the preconditioner can then be amortized \cite{melo2025variational}.
\end{itemize}
In \cref{sec:preconditioner}, we propose a preconditioner that is adapted to the structure of the Lindbladian, and in \cref{sec:low_lying} show it can be used to compute the low-lying spectrum of $\LL$.

\paragraph{Simulating long-time dynamics.}
A last approach, often used in practice, is to obtain the steady state by evolving an initial state $\rho_0\geq 0, \Tr{\rho_0}=1$ for a sufficiently long time $T$: $\rho_\infty = \lim_{T \to \infty} e^{\LL T} \rho_0$. This requires integrating the master equation, which can be computationally expensive for stiff systems (as they require small time-steps or expensive solvers) and/or those with small spectral gaps (as the convergence with time $T$ is slow).

\subsection{New method based on an auxiliary CPTP map}
\label{subsec:fixed_point_method}

We propose a new method that reformulates the steady state problem as that of finding the fixed point of a specific CPTP map $\Phi = -\mathcal{K}\mathcal{S}^{-1}$, defined whenever $\SS$ is invertible. Then, as a CPTP map, $\Phi$ has its spectrum contained in the closed unit disk. An interesting result is that under the assumptions of \cref{thm:steady-state-eigenvalue} below, $1$ is a simple eigenvalue of maximal modulus, whose eigenvector $\xi_\infty$ yields the steady state through $\rho_\infty = -\SS^{-1}\xi_\infty / \Tr(-\SS^{-1}\xi_\infty)$, making Arnoldi iterations particularly efficient. Further, $\Phi$ can be efficiently applied at cost $O(n^3)$.

The first point is made precise thanks to the following theorem:
\begin{theorem}[Steady state as fixed point of a linear operator]\label{thm:steady-state-eigenvalue}
    Assume $\dim\ker(\LL)=1$ and recall that $G=-iH - \frac{1}{2}\sum_{j=1}^d L_j^\dagger L_j$. Then either
    \begin{itemize}
        \item $\vecspec(G)\cap i\mathbb{R}\neq \emptyset$. In this case there exists an eigenstate of $H$ in the kernel of all jump operators, and the kernel of $\LL$ is the span of the density matrix associated to that eigenstate.
        \item $\vecspec(G)\cap i\mathbb{R}=\emptyset$. Then $\SS$ is invertible and $\Phi\coloneq -\KK\SS^{-1}$ is a CPTP map with a unique invariant state $\xi_\infty$. Besides, there exists a positive constant $c$ such that $\rho_\infty= -c \SS^{-1}\xi_\infty$.
    \end{itemize}
\end{theorem}
Loosely speaking, either the steady state is a pure state, in which case it is easy to find by diagonalizing the no-jump generator $G$ of \cref{eq:G}, or it is a mixed state and can be found as the fixed point of a CPTP map that is not computationally expensive to evaluate. Let us now give the proof of this result.

\begin{proof}
    Assume $\vecspec(G)\cap i\mathbb{R}\neq\emptyset$. Let $-i\lambda$, with $\lambda\in\mathbb{R}$, be an eigenvalue of $G$ with normalized eigenvector $\ket{\psi_\lambda}$, i.e.
    $G\ket{\psi_\lambda}=-i\lambda\ket{\psi_\lambda}$. Taking the scalar product with $\bra{\psi_\lambda}$ gives
    \begin{equation}
        \bra{\psi_\lambda}G\ket{\psi_\lambda} = -i\bra{\psi_\lambda}H\ket{\psi_\lambda} - \frac{1}{2}\sum_j\|L_j\psi_\lambda\|^2 = -i\lambda.
    \end{equation}
    Taking the real part yields $\sum_j\|L_j\psi_\lambda\|^2 = 0$, hence $\ket{\psi_\lambda}\in\bigcap_j\ker(L_j)$; then $G\ket{\psi_\lambda} = -iH\ket{\psi_\lambda}$ shows that $\ket{\psi_\lambda}$ is an eigenvector of $H$ with eigenvalue $\lambda$. A direct computation then gives $\LL(\ket{\psi_\lambda}\bra{\psi_\lambda}) = 0$, and since $\dim\ker(\LL)=1$, we get $\ker\LL = \mathbb{C}\ket{\psi_\lambda}\bra{\psi_\lambda}$.

    Assume now $\vecspec(G)\cap i\mathbb{R}=\emptyset$; since the spectrum of $G$ lies in the closed left half-plane, this means $\Re(\vecspec(G))<0$. Then $\vecspec(G) \cap \vecspec(-G^\dag) = \emptyset$, so the Sylvester equation $GX+XG^\dag=Y$ has a unique solution for every $Y$ given by \cref{eq:lyapunov_solution}. Hence, $\SS$ is invertible. The map $-\KK\SS^{-1}$ is completely positive as the composition of the two completely positive maps $-\SS^{-1}$ (see \cref{eq:lyapunov_solution}) and $\KK$. The trace-preserving property is checked in \cref{app:properties_phi}.

    Let us now make the connection between the fixed points of $\Phi=-\KK\SS^{-1}$ and $\LL$. For any matrix $\rho$ we have
    \begin{align*}
        \rho\in\ker\LL
         & \iff \KK(\rho)+\SS(\rho)=0                                                  \\
         & \iff -\SS^{-1}\KK(\rho)=\rho                                                \\
         & \iff \rho\in\ker(\mathcal{M}-\Id), \qquad \mathcal{M}\coloneq -\SS^{-1}\KK,
    \end{align*}
    so that $\ker(\mathcal{M}-\Id)=\mathbb{C}\rho_\infty$ is one-dimensional. The operators $\Phi=-\KK\SS^{-1}$ and $\mathcal{M}=-\SS^{-1}\KK$ are similar, $\mathcal{M}=\SS^{-1}\Phi\SS$, and their fixed points are in one-to-one correspondence:
    \begin{itemize}
        \item if $\mathcal{M}\rho=\rho$, then $\Phi(\KK\rho)=-\KK\SS^{-1}\KK\rho=-\KK(-\rho)=\KK\rho$, with $\KK\rho\neq0$ since $\SS^{-1}\KK\rho=-\rho\neq0$;
        \item if $\Phi\xi=\xi$, then $\mathcal{M}(-\SS^{-1}\xi)=\SS^{-1}\KK\SS^{-1}\xi=-\SS^{-1}\xi$, using $\KK\SS^{-1}\xi=-\Phi\xi=-\xi$.
    \end{itemize}
    Hence $\ker(\Phi-\Id)$ is one-dimensional as well, and the injective map $\xi\mapsto-\SS^{-1}\xi$ maps it onto $\ker(\mathcal{M}-\Id)=\ker\LL$, which proves $\ker\LL=-\SS^{-1}\ker(\Phi-\Id)$. Moreover, since $\Phi$ is CPTP, it admits an invariant steady state. Besides, as $\ker(\Phi-\Id)$ is one-dimensional, this invariant state $\xi_\infty$ is unique. Finally, $-\SS^{-1}\xi_\infty$ is positive semi-definite (by complete positivity of $-\SS^{-1}$) and nonzero, so its trace is positive, and $\rho_\infty=-c\,\SS^{-1}\xi_\infty$ with $c=1/\Tr(-\SS^{-1}\xi_\infty)>0$.
\end{proof}

\begin{remark}[Correspondence of kernels without uniqueness]
    \label{rk:kernel-correspondence}
    The assumption $\dim\ker\LL = 1$ plays no role in the correspondence itself. Assuming only that $\SS$ is invertible, the proof above shows that $\ker\LL = \ker(\mathcal{M}-\Id)$, with $\mathcal{M} = -\SS^{-1}\KK$, and that $\xi \mapsto -\SS^{-1}\xi$ maps $\ker(\Phi-\Id)$ bijectively onto it, so that
    \begin{equation}
        \ker \LL = -\SS^{-1}\ker(\Phi-\Id)
    \end{equation}
    holds in general; in particular $\dim\ker(\Phi-\Id) = \dim\ker\LL$. Uniqueness of the steady state is only used to conclude that these two kernels are one-dimensional.
\end{remark}

\paragraph{Computing the fixed point of $\Phi$.}
\Cref{thm:steady-state-eigenvalue} allows us to compute the steady state $\rho_\infty$ by finding the fixed point of the CPTP map $\Phi=-\KK\SS^{-1}$. We recall that any CPTP map has its spectrum contained in the closed unit disk, and all eigenvalues of modulus one are semisimple \cite[Theorem 6.6]{wolf_guided_tour_2012}. Since $\ker(\Phi - \Id)$ is one-dimensional and $1$ is a semisimple eigenvalue, $1$ is a simple eigenvalue of $\Phi$. Moreover it is of maximal modulus, which is precisely the regime in which Krylov methods are effective --- in contrast with $\LL$, whose eigenvalue of interest is the smallest in magnitude. Simplicity does not, however, preclude \emph{other} eigenvalues of modulus one, and \cref{rk:peripheral} exhibits a system where $\Phi$ has three: a power iteration on $\Phi$ would then fail to converge. Arnoldi is less exposed, because it builds the whole Krylov subspace, in which every peripheral mode is represented, and we retain the Ritz vector whose Ritz value is closest to $1$ rather than the dominant one (\cref{alg:steady-state-arnoldi}). In \cref{fig:spectrum_L_and_Phi}, we plot the spectrum of $\LL$ and $\Phi$ for a random Lindbladian (see \cref{sec:systems_studied} for details).

\begin{figure}[htb]
    \centering
    \includegraphics[width=\linewidth]{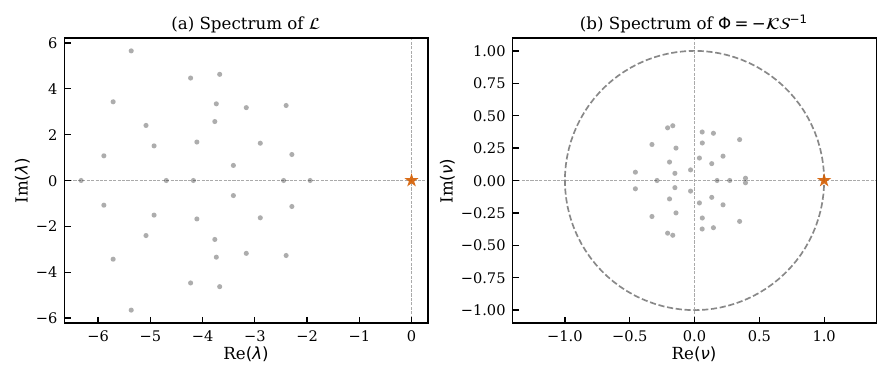}
    \caption{Spectra of a random dense Lindbladian (see \cref{sec:systems_studied}) for a Hilbert space of dimension $n=6$.
        (a)~Spectrum of $\LL$; even for low-dimensional systems, the large spread of eigenvalues reflects the poor conditioning of $\LL$.
        (b)~Spectrum of $\Phi = -\mathcal{K}\mathcal{S}^{-1}$; all eigenvalues lie within the unit circle (dashed).
        In both panels, the red star marks the eigenvalue corresponding to the steady state.}
    \label{fig:spectrum_L_and_Phi}
\end{figure}

\begin{remark}[Peripheral spectra of $\Phi$]
    \label{rk:peripheral}
    For CPTP maps arriving as solutions of a Lindblad equation, assuming uniqueness of the steady state ensures that there cannot be any non-trivial peripheral eigenvalues \cite[Proposition 7.5]{wolf_guided_tour_2012}, that is eigenvalues of modulus one different from $1$. This is not the case for $\Phi$, as illustrated by the following example: consider a three level system with jump operators $L_1=\ket{0}\bra{1}, L_2=\ket{1}\bra{2}$ and $L_3=\ket{2}\bra{0}$, and no Hamiltonian. The unique steady state is the maximally mixed state $\rho_\infty=\frac{\Id}{3}$, but $\Phi$ has three peripheral eigenvalues $1, e^{2\pi i/3}, e^{-2\pi i/3}$. This situation is completely analogous to that of periodic classical Markov chains, whose transition matrices have several eigenvalues on the unit circle despite the stationary distribution being unique.
\end{remark}

\begin{remark}[Relation with a numerical integration scheme]
    The very recent paper \cite{doi:10.1137/24M1690795} introduces a first-order, CPTP time-discretization of the Lindblad equation, given by:
    \begin{align}
        \label{eq:scheme_hao}
        \rho_{n+1} & =e^{\delta t G}\rho_n e^{\delta t G^\dag} + \sum_j \int_0^{\delta t} L_j e^{(\delta t-s)G} \rho_n e^{(\delta t-s)G^\dag} L_j^\dag ds \\
                   & =e^{\delta t G}\rho_n e^{\delta t G^\dag} + \Phi(\rho_n- e^{\delta t G}\rho_n e^{\delta t G^\dag}),
    \end{align}
    where the second equality holds whenever $\SS$ is invertible. Since $\Phi$ is CPTP, the integrator above is CPTP as well. Note that in the limit $\delta t\to+\infty$, the scheme reduces to $\rho_{n+1} = \Phi(\rho_n)$. It is thus remarkable that the steady state of this scheme in the limit $\delta t\to+\infty$, despite not being equal to the steady state of the continuous Lindblad dynamics, is directly related to it via \cref{thm:steady-state-eigenvalue}.

\end{remark}
\subsection{Numerical benchmarks}
\label{sec:benchmark_steady_state}
\subsubsection{Description of the benchmarks}
\label{sec:systems_studied}

In order to validate the theoretical efficiency of the methods introduced in \cref{thm:steady-state-eigenvalue}, we benchmark them on two classes of Lindbladian systems with markedly different structural properties:
\begin{itemize}
    \item A \emph{dense random Lindbladian}, designed to assess robustness and performance on generic, unstructured generators: a random Hermitian Hamiltonian together with three dense jump operators with i.i.d.\ complex-normal entries. Details are provided in \cref{app:dense_gen}.
    \item A \emph{physically motivated sparse Lindbladian} describing memory--buffer cat states, representative of hardware-induced dynamics encountered in the \textit{Alice \& Bob}\footnote{\url{https://alice-bob.com/}} quantum processors. All four jump operators and the Hamiltonian are sparse, and the system exhibits a small spectral gap. Details are provided in \cref{sec:cat}.
\end{itemize}

The tests are performed on a single compute node equipped with a CPU Intel Xeon Platinum 8481C and 70\,GB of RAM; GPU benchmarks were conducted on an Nvidia H100 with 80\,GB of memory. All computations use double-precision floating-point arithmetic.

For iterative algorithms, convergence is declared when all entries of $\LL(\rho)$ fall below $10^{-8}$ in absolute value, i.e. $\norm{\LL(\rho)}_{\max} < 10^{-8}$, after ensuring $\rho$ is Hermitian with trace 1. Each simulation is terminated early if the runtime exceeds $100\,\mathrm{s}$ or if available memory is exhausted.
The reported quantity is the total wall-clock time as a function of the Hilbert space dimension~$n$. The methods that need the Liouvillian superoperator explicitly represent it as a matrix of size $n^2 \times n^2$, possibly in a sparse data format for the cat qubit example.

\subsubsection{Benchmark}

\begin{figure}[tb]
    \centering
    \includegraphics[width=\linewidth]{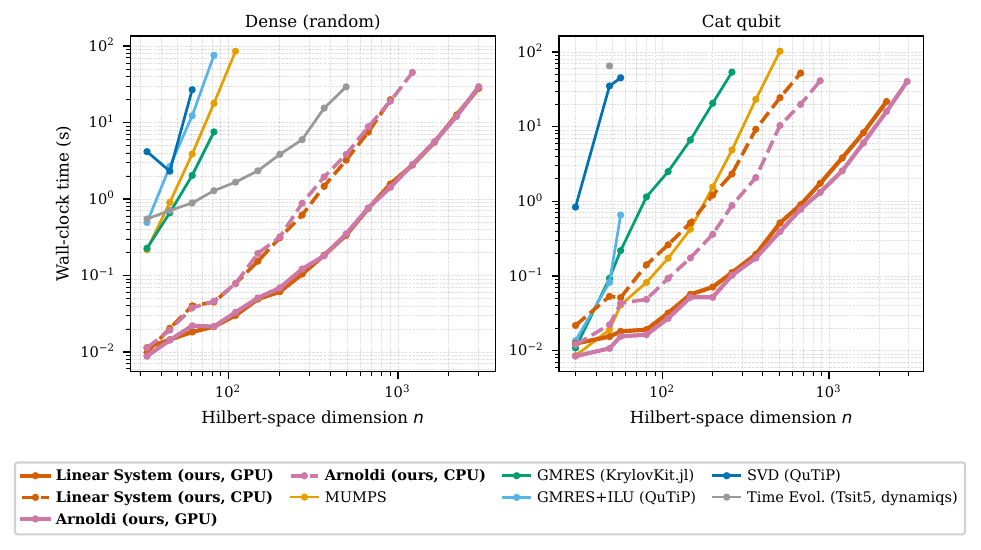}
    \caption{Runtime of steady-state solvers versus Hilbert-space dimension $n$, for the dense random Lindbladian (left) and the memory--buffer cat qubit (right). Our two methods (Arnoldi on $\Phi$ and preconditioned GMRES) beat all competitors on CPU, and further gain a $\times 50$ speedup on the H100 GPU; the competitors (QuTiP SVD and ILU-GMRES, MUMPS, Krylov.jl GMRES) run on the Xeon CPU, and the time-evolution baseline on GPU for the cat family. See \cref{sec:details_num_benchmark} in the Appendix for the exact configurations.}
    \label{fig:benchmark-steadystate}
\end{figure}

The results of the benchmarks are reported in \cref{fig:benchmark-steadystate}. For small Hilbert space dimensions, several methods achieve comparable runtimes. As the dimension increases, however, the asymptotic scaling of the different approaches becomes clearly visible. The direct SVD-based solver quickly becomes prohibitively expensive, exhibiting the expected $O(n^6)$ scaling associated with dense factorizations. The iterative baselines (the ILU-preconditioned GMRES of QuTiP and the unpreconditioned Krylov.jl GMRES) also scale poorly and encounter memory limitations at larger dimensions. The sparse direct solver MUMPS performs well despite running on CPU, but it is not available on GPU and cannot be applied to large, dense operators arising in the random Lindbladian benchmark.
Finally, the ODE integration approach is significantly slower in absolute runtime, but displays favorable scaling. Its poor performance on the cat qubit benchmark is due to the difficulty of reaching the required tolerance, even at large propagation times\footnote{This is unsurprising: dissipative cat qubits exhibit a strong metastability, due to the exponential suppression of the bit-flip.}.

In contrast, the Arnoldi method applied to $\Phi$ and our preconditioned GMRES solver described in \cref{sec:preconditioner} exhibit excellent scaling, outperform the other methods, and show a further consistent speedup when run on GPU.

\section{A new preconditioner to compute the resolvent}
\label{sec:preconditioner}

Unfortunately, the previous method cannot be straightforwardly adapted to compute the low-lying spectrum of $\LL$. Indeed, the eigenvectors and eigenvalues of $\Phi=-\KK\SS^{-1}$ are not related in a simple way to those of $\LL$ for eigenvalues of $\Phi$ different from $1$. Hence, in this section we propose to leverage the cheap inversion of $\SS$ to build a preconditioner for the shifted Lindbladian $\lambda - \LL$ for some $\lambda > 0$. We present in \cref{sec:applications} three applications of this preconditioner: (i) the computation of the steady state as a linear system (an alternative to the method presented in \cref{subsec:fixed_point_method}), (ii) the computation of the low-lying spectrum of $\LL$ through a shift-invert method, and (iii) the implicit time integration of the Lindblad equation.

\subsection{A preconditioner from the no-jump resolvent}
\label{sec:precond_def}

\paragraph{Reminder: Iterative methods and preconditioning.}
Iterative methods are designed to solve approximately a linear system. Given an invertible complex matrix $A$ and a vector $b$, they look for an approximate solution of $A x = b$ in the Krylov space
\begin{equation}\label{eq:intro-krylov-space}
    \mathcal{V}_p = \vspan \cb{b, A b, \dots, A^{p-1} b},
\end{equation}
that is, they approximate $A^{-1}b$ by $q(A)b$ for some polynomial $q$ of degree less than $p$. Such an approximation is exact as soon as $p$ reaches the degree of the minimal polynomial of $A$ (at most $n^2$ here), but in practice far fewer iterations are needed for a given tolerance. The convergence speed depends on the conditioning of the problem and, for non-normal operators, on the clustering of the eigenvalues away from the origin \cite{saad_gmres,gmres1986generalized}; for non-normal operators eigenvalues alone determine nothing and convergence can be arbitrarily bad \cite{Greenbaum2012Feb,Arioli1998Dec}; this is why iterative methods are seldom used without a preconditioner, that is an invertible operator $P$ which is inexpensive to apply and which approximates $A^{-1}$. We distinguish left preconditioning, which amounts to solving $P A x = P b$, from right preconditioning, which uses
\begin{equation}
    A P y = b, ~~ x = P y.
\end{equation}
While several Krylov-based iterative solver algorithms can be used, we restrict ourselves to GMRES \cite{gmres1986generalized} in our numerical experiments.

\paragraph{The no-jump resolvent.}
We are interested in solving linear systems of the form $(\lambda - \LL)(\rho)=b$ for $\lambda >0$, or the rank-one deflated variant of it $(\lambda -\widetilde \LL_\eta)(\rho)=b$ for $\lambda \geq 0$. The candidate preconditioner is the no-jump resolvent
\begin{equation}\label{eq:no-jump-resolvent-precond}
    \RR_\lambda^\SS \coloneq (\lambda - \SS)^{-1},
\end{equation}
whose basic properties we now record. Throughout, we write
\begin{equation}\label{eq:G-lambda}
    G_\lambda \coloneq G - \tfrac{\lambda}{2} \Id.
\end{equation}
For $\lambda = 0$, we recover the no-jump generator $G$ and the no-jump resolvent $\RR_0^\SS = -\SS^{-1}$ is defined under the assumption that $\SS$ is invertible, see \cref{thm:steady-state-eigenvalue}. For $\lambda > 0$, the operator
\begin{align}
    \label{eq:expr_SS}
    (\lambda - \SS)(X)= -(G_\lambda X + X G_\lambda^\dagger)
\end{align}
is always invertible, as $G_\lambda$ has spectrum in the open left half-plane. The computation of $\RR_\lambda^\SS$ thus amounts again to solving a continuous Lyapunov equation.

\begin{theorem}[$\KK\RR_\lambda^\SS$ is a strict contraction]\label{thm:contraction}
    Let $\lambda > 0$ and let $\norm{\cdot}$ denote the norm induced by the trace norm. Then,
    \begin{equation}\label{eq:contraction-bound}
        \norm{\KK \RR_\lambda^\SS} < 1 .
    \end{equation}
\end{theorem}

\begin{proof}
    First note that $\RR_\lambda^\SS$ is completely positive, and hence so is $\KK \RR_\lambda^\SS$. The induced trace norm of a completely positive map $\Psi$ is attained on density matrices: by the Russo--Dye theorem $\norm{\Psi} = \norm{\Psi^*}_{\infty \to \infty} = \norm{\Psi^*(\Id)}_\infty$, and $\Psi^*(\Id) \succeq 0$, so that $\norm{\Psi} = \sup_{\rho \succeq 0, \Tr\rho = 1} \Tr \Psi(\rho)$ (see \eg \cite[Chapter 6]{wolf_guided_tour_2012}). As a consequence,
    \begin{equation}
        \norm{\KK \RR_\lambda^\SS} = \sup_{\rho \succeq 0, \Tr \rho = 1} \Tr \KK \RR_\lambda^\SS \rho.
    \end{equation}
    Using the fact that $\Tr \LL \sigma = 0$ for any $\sigma$, we have $\Tr \KK \sigma = - \Tr \SS \sigma$. Hence,
    \begin{equation}\label{eq:proof-resolvent-sum-objective}
        \norm{\KK \RR_\lambda^\SS} = \sup_{\rho \succeq 0, \Tr \rho = 1} -\Tr \SS \RR_\lambda^\SS \rho.
    \end{equation}
    Writing $- \SS \RR_\lambda^\SS = (\lambda - \SS) \RR_\lambda^\SS - \lambda \RR_\lambda^\SS = \Id - \lambda \RR_\lambda^\SS$,
    the trace becomes
    \begin{equation}\label{eq:proof-resolvent-sum-trace}
        -\Tr \SS \RR_\lambda^\SS \rho = 1 - \lambda \Tr \RR_\lambda^\SS (\rho).
    \end{equation}
    For any normalized state $\ket{\psi}$, $\RR_\lambda^{\SS}$ in integral form gives $\Tr (\RR_\lambda^{\SS}(\ket \psi \bra \psi)) = \int_{0}^{\infty} \norm{e^{G_\lambda t} \ket \psi}^2 \dd t > 0$, and by linearity $\Tr \RR_\lambda^\SS (\rho) > 0$ for every density matrix $\rho$.
    Since the set of density matrices is compact in finite dimension, the infimum
    $\delta_\lambda \coloneq \min_{\rho \succeq 0, \Tr\rho = 1} \Tr \RR_\lambda^\SS \rho$
    is attained and is positive. Combining with \cref{eq:proof-resolvent-sum-objective,eq:proof-resolvent-sum-trace} yields $\norm{\KK \RR_\lambda^\SS} = 1 - \lambda \delta_\lambda < 1$.
\end{proof}

\paragraph{Right preconditioning by the no-jump resolvent.}
We can now justify the use of $\RR_\lambda^\SS$ as a right preconditioner for $\lambda - \LL$. Applying $\lambda - \LL$ to the left of $\RR_\lambda^\SS$ and using $\LL = \SS + \KK$ gives
\begin{equation}\label{eq:right-precond-identity}
    (\lambda - \LL)\,\RR_\lambda^\SS = \Id - \KK\RR_\lambda^\SS .
\end{equation}
The quality of the preconditioner is thus governed by $\norm{\KK\RR_\lambda^\SS}$, which \cref{thm:contraction} bounds strictly below $1$. Consequently the spectrum of the preconditioned operator $\Id - \KK\RR_\lambda^\SS$ is contained in the disk of center $1$ and radius $\norm{\KK\RR_\lambda^\SS} < 1$; in particular it is bounded away from the origin. For a non-normal operator, eigenvalue location alone does not, in full generality, determine the behavior of GMRES. Here, however, \cref{thm:contraction} controls a norm which turns into an actual convergence rate.

\begin{remark}[Geometric convergence of the preconditioned GMRES]
    \label{rk:gmres-rate}
    Write $\mathcal{N} \coloneq \KK\RR_\lambda^\SS$, so that the preconditioned system reads $(\Id - \mathcal{N})y = b$. Its $k$-th Krylov space $\mathcal{V}_k$ contains the truncated Neumann iterate $\sum_{j<k}\mathcal{N}^j b$, of residual $\mathcal{N}^k b$, so residual minimality and $\norm{\cdot}_2 \leq \norm{\cdot}_1 \leq \sqrt n\,\norm{\cdot}_2$ on $\CC^{n\times n}$ give the following bound on the $k$-th GMRES residual $r_k$:
    \begin{equation}\label{eq:gmres-rate}
        \norm{r_k}_2 = \min_{y \in \mathcal{V}_k} \norm{b - (\Id - \mathcal{N})y}_2 \;\leq\; \norm{\mathcal{N}^k b}_2 \;\leq\; \sqrt{n}\,\norm{\mathcal{N}}^k \norm{b}_2 :
    \end{equation}
    by \cref{thm:contraction} the preconditioned GMRES converges at least geometrically with ratio $\norm{\mathcal{N}} < 1$, through a bound that uses no spectral information and hence is insensitive to non-normality.
\end{remark}

Note that left preconditioning is less natural, as we would obtain $\Id - \RR_\lambda^\SS\KK$. Interestingly, while $\RR_\lambda^\SS\KK$ and $\KK \RR_\lambda^\SS$ are similar (thus have the same spectrum), contrary to $\KK \RR_\lambda^\SS$, the norm of $\RR_\lambda^\SS\KK$ can be arbitrarily large.\footnote{Take $n = 2$, $H = 0$ and a single jump operator $L = \sqrt\gamma\,\ket{0}\!\bra{1}$, so that $G = -\tfrac{\gamma}{2}\ket{1}\!\bra{1}$ is diagonal in the computational basis. A direct computation gives $\RR_\lambda^\SS\KK(\rho) = \tfrac{\gamma}{\lambda}\,\rho_{11}\,\ket{0}\!\bra{0}$ and $\KK\RR_\lambda^\SS(\rho) = \tfrac{\gamma}{\gamma+\lambda}\,\rho_{11}\,\ket{0}\!\bra{0}$, whence $\norm{\RR_\lambda^\SS\KK} = \gamma/\lambda$, which is unbounded as $\lambda \to 0$, while $\norm{\KK\RR_\lambda^\SS} = \gamma/(\gamma+\lambda) < 1$.}

\begin{remark}[Loss of the contraction at zero shift]
    \label{rk:zero-shift}
    \Cref{thm:contraction} requires $\lambda > 0$, and its conclusion fails exactly at the endpoint. Indeed, if $\SS$ is invertible, then $\RR_0^\SS = -\SS^{-1}$ and the error operator of \cref{eq:right-precond-identity} becomes the CPTP map $\Phi = -\KK\SS^{-1}$ of \cref{thm:steady-state-eigenvalue}, so that $\norm{\KK\RR_0^\SS} = \norm{\Phi} = 1$: the contraction is no longer strict.

    This is not a mere technicality, since $\lambda = 0$ is precisely the shift at which two of our three applications operate: the steady-state solver of \cref{sec:steady_state_linear_solve} and the shift-invert of \cref{sec:low_lying} both precondition there, albeit on the deflated operator $\widetilde\LL_\eta$ rather than on $\LL$ itself. Deflation does make the linear system nonsingular, but we have not established that it restores a strict contraction, and the peripheral eigenvalues of $\Phi$ exhibited in \cref{rk:peripheral} might prevent this. So we claim no convergence rate for these two solvers. The gap is in the theory only: in all our experiments GMRES converges in a few iterations, as the benchmarks of \cref{sec:applications} show.
\end{remark}

\subsection{Link with the resolvent formula}
\label{sec:resolvent}

The estimate of \cref{thm:contraction} has a second consequence, of a more theoretical nature: it yields an explicit series representation of the resolvent of the Lindbladian for $\lambda > 0$,
\begin{equation*}
    \RR_\lambda^\LL \coloneq (\lambda - \LL)^{-1}.
\end{equation*}

\begin{corollary}[Resolvent of the Lindbladian]\label{thm:resolvent}
    Assume $\lambda > 0$. Then $\Id - \KK \RR_\lambda^\SS$ is invertible, $\lambda$ belongs to the resolvent set of $\LL$, and $\RR_\lambda^\LL$ is given by the convergent series
    \begin{equation}\label{eq:thm-main}
        \RR_\lambda^\LL = \RR_\lambda^\SS \sum_{k=0}^\infty (\KK \RR_\lambda^\SS)^k.
    \end{equation}
\end{corollary}

\begin{proof}
    By \cref{thm:contraction}, $\norm{\KK \RR_\lambda^\SS} < 1$, so the Neumann series $\sum_{k\geq 0} (\KK \RR_\lambda^\SS)^k$ converges, and its sum is the inverse of $\Id - \KK \RR_\lambda^\SS$. Multiplying \cref{eq:right-precond-identity} on the right by that inverse gives
    \begin{equation}\label{eq:resolvent-right-inverse}
        (\lambda - \LL)\; \RR_\lambda^\SS \p{\Id - \KK \RR_\lambda^\SS}^{-1} = \Id ,
    \end{equation}
    so $\lambda - \LL$ admits a right inverse. Hence, $\lambda$ belongs to the resolvent set of $\LL$, and $\RR_\lambda^\LL = \RR_\lambda^\SS(\Id - \KK \RR_\lambda^\SS)^{-1}$, which is \cref{eq:thm-main}.
\end{proof}

In fact, \Cref{eq:thm-main} was already proved in the 1990s in a more general infinite-dimensional framework with unbounded generators. For this more complex setting, uniqueness of the solution of the Lindblad equation is not guaranteed, and is equivalent to the so-called conservativity of the minimal semigroup. It turns out that the resolvent of this minimal semigroup admits a representation\footnote{More precisely, one usually works with the dual of the semigroup (\ie in the Heisenberg picture), thus the formula is stated for the dual and the series converges only in the strong operator topology} given by \cref{eq:thm-main}. We refer to the works \cite{chebotarevSufficientConditionsConservativity1993,chebotarevSufficientConditionsConservativity1998}.

Beyond its theoretical interest, the series \cref{eq:thm-main} is closely tied to the iterative solver of \cref{sec:precond_def}: truncating it and running a preconditioned Krylov method explore the same space, the latter selecting the residual-minimizing element of that space. We come back to this connection, and to the steady state method it motivates, in \cref{sec:steady_state_linear_solve}.

\subsection{Numerical benchmarks}
\label{sec:precond_benchmarks}

We benchmark our resolvent solver on the same two families of generators and with the same convergence criterion and hardware as in \cref{sec:benchmark_steady_state}. The linear system to solve is $(\lambda - \LL)(\rho) = b$. The right-hand side is the maximally mixed state $b = \Id/n$, and we sweep three shifts $\lambda = \lambda_{\max}\cdot 10^{\{0,-1,-2\}}$, where $\lambda_{\max}$ estimates the spectral radius of $\LL$ and is computed once, off the clock. Our solver is the right-preconditioned recycled GMRES of \cref{sec:precond_def}; the competitors are a dense $LU$ solve, MUMPS, and (un)preconditioned SciPy GMRES, all detailed in \cref{sec:details_num_benchmark}.

\begin{figure}[tb]
    \centering
    \includegraphics[width=\linewidth]{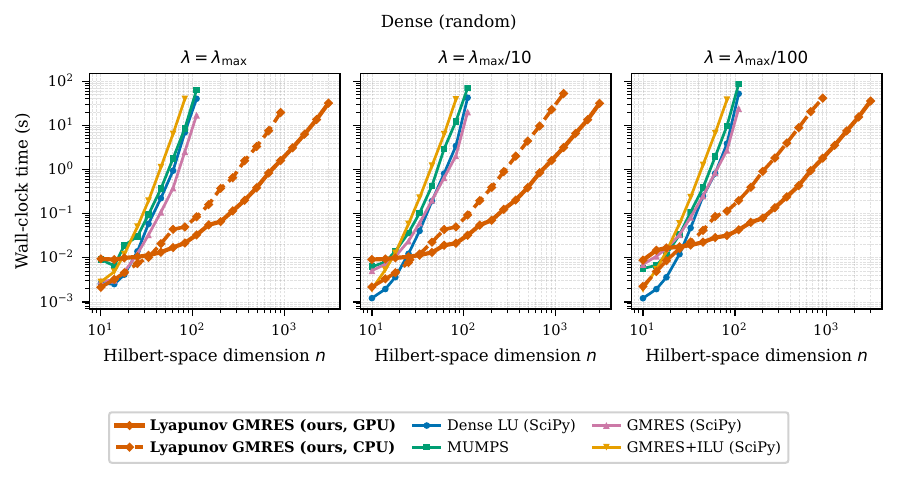}\\[4pt]
    \includegraphics[width=\linewidth]{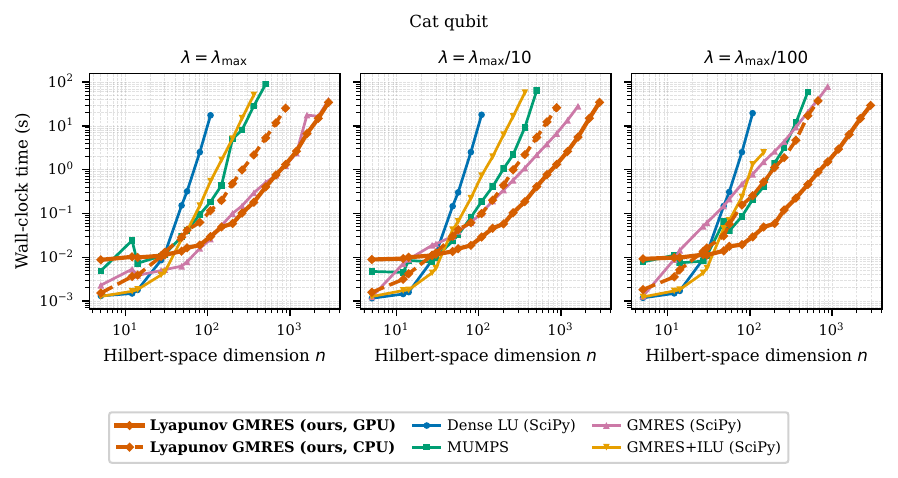}
    \caption{Runtime of resolvent solvers for $(\lambda - \LL)(\rho) = \Id/n$ versus Hilbert-space dimension $n$, at the three shifts $\lambda = \lambda_{\max},\ \lambda_{\max}/10,\ \lambda_{\max}/100$ (columns), for the dense random Lindbladian (top) and the memory--buffer cat qubit (bottom). Our preconditioned GMRES is shown on CPU (dashed) and on the H100 GPU (solid); the competitors (dense $LU$, MUMPS, GMRES, GMRES\,+\,ILU) run on the Xeon CPU. The gain over the competitors is small at $\lambda = \lambda_{\max}$ and grows as $\lambda \to 0$, and is largest for the dense family and on GPU. See \cref{sec:details_num_benchmark} for the exact configurations.}
    \label{fig:benchmark-resolvent}
\end{figure}

The results are shown in \cref{fig:benchmark-resolvent}. At the largest shift $\lambda = \lambda_{\max}$ the operator $\lambda - \LL$ is well-conditioned, so even unpreconditioned GMRES converges in a handful of iterations and our preconditioner brings little gain (no measurable speedup on the sparse cat family). As $\lambda$ decreases toward the singular limit $\lambda \to 0$, the unpreconditioned iterations degrade while $\RR_\lambda^\SS$ keeps the preconditioned spectrum clustered near $1$ (\cref{thm:contraction}), and the advantage of our method becomes clear. We gain a factor $\sim\!50$ at $\lambda_{\max}/100$ on the cat family. The gain is largest on the dense systems, where the competitors are bound by the factorization of the $n^2\times n^2$ Liouvillian and drop out near $n \approx 100$, whereas our matrix-free solver built from the $n\times n$ generator reaches $n = 3000$. These improvements are again amplified on GPU.

\section{Applications}
\label{sec:applications}

The preconditioner of \cref{sec:preconditioner} unlocks several applications. We detail three: an alternative to the fixed-point method of \cref{sec:fixed_point} for the computation of the steady state, the computation of the low-lying spectrum of the Lindbladian through a shift-invert transformation, and an implicit time integrator for stiff master equations.

\subsection{Steady state as a preconditioned linear solve}
\label{sec:steady_state_linear_solve}

Recall from \cref{sec:fixed_point} that the steady state solves the deflated linear system $\widetilde\LL_\eta\,\rho = \eta\,\Id$ of \cref{eq:steady_state_linear_problem}, equivalently its $\lambda = 0$ form $(0 - \widetilde\LL_\eta)\,\rho = -\eta\,\Id$ in the convention $(\lambda - \widetilde\LL_\eta)\rho = b$ of \cref{sec:precond_def}, which is the one we precondition. Assuming the no-jump generator $\SS$ to be invertible, we solve this system with the right-preconditioned GMRES of \cref{sec:precond_def}, taking as preconditioner the no-jump resolvent at $\lambda = 0$, into which the rank-one deflation is folded by the Sherman--Morrison formula. Refer to  \cref{eq:ShermanMorr} and \cref{alg:precond,alg:steady-state} for a precise definition. This provides an alternative to the fixed-point iteration of \cref{sec:fixed_point}, with which it is on par in \cref{fig:benchmark-steadystate}.

The series representation \cref{eq:thm-main} explains why this is natural. Assume $\SS$ is invertible and set formally $\lambda = 0$ in \cref{eq:thm-main} (the series then need not converge, since $\norm{\KK \RR_0^\SS} = \norm{\Phi} = 1$): truncating it applied to some $\rho_0$ amounts to spanning $\SS^{-1} \vspan \cb{(\KK \SS^{-1})^k \rho_0}$, which is, up to the $\SS^{-1}$ transform, exactly the Krylov space $\vspan \cb{(\LL \SS^{-1})^k \rho_0}$ obtained by right preconditioning an equation of the form $\LL \rho = \rho_0$ (indeed $\LL\SS^{-1} = \Id + \KK\SS^{-1}$, and both operators generate the same Krylov spaces). Truncating the resolvent series and running a preconditioned Krylov method therefore explore the same space; the Krylov method additionally selects the residual-minimizing element of that space, whereas the truncated series takes the one prescribed by the Neumann expansion.

\begin{remark}[Differentiating the steady state]
    \label{rk:differentiating-steady-state}
    Because the steady state is obtained as the solution of a linear system, it is moreover differentiable through the adjoint method. Write the system as $\widetilde\LL_\eta(\theta)\,\rho_\infty = \eta\,\Id$, with $\theta$ collecting the parameters of $H$ and the $L_j$, and let $\mathcal{J}(\rho_\infty)$ be a smooth scalar figure of merit. Differentiating the system implicitly gives $\nabla_\theta \mathcal{J} = -\dotprod{w}{(\partial_\theta \widetilde\LL_\eta)\,\rho_\infty}$, where the adjoint state $w$ solves the single linear system $\widetilde\LL_\eta^*\,w = \nabla_\rho \mathcal{J}$. The whole gradient thus costs one additional solve, to which the same no-jump preconditioner applies since $\lambda - \SS^*$ is again a Lyapunov operator. Crucially, one differentiates the exact solution rather than the GMRES iterations, so the cost and the accuracy of the gradient do not depend on the number of Krylov steps. This is known as implicit differentiation \cite{Krantz}.
\end{remark}

\subsection{Low-lying spectrum}
\label{sec:low_lying}

The low-lying spectrum of $\LL$ is the set of its eigenvalues of smallest magnitude; these determine the dissipative gap and the slowest relaxation rates of the system. As in \cref{sec:fixed_point}, we assume throughout this subsection that $\dim \ker \LL = 1$, and we place ourselves in the second case of \cref{thm:steady-state-eigenvalue}, in which the no-jump generator $\SS$ is invertible.

\subsubsection{State of the art}
To the best of our knowledge, the methods available for computing the low-lying spectrum fall into two categories, which we described in \cref{sec:fixed_point}.

\paragraph{Dense diagonalization.}
One forms the Liouvillian as an $n^2 \times n^2$ matrix and diagonalizes it, at a cost $\mathcal{O}(n^6)$ that quickly becomes intractable.

\paragraph{Iterative diagonalization.}
One constructs a Krylov subspace to directly find the eigenvalues of $\mathcal{L}$ with lowest magnitude. This is very inefficient, as the Arnoldi method works well for eigenvalues with largest magnitude, and typically requires reconstructing the full space.

An interesting option to make the iterative approach viable is the Arnoldi--Lindblad method of \cite{Minganti2022Feb}: rather than diagonalizing $\LL$, it applies Arnoldi to the propagator $e^{\LL t}$, whose action is evaluated by numerically integrating the master equation, without ever forming the $n^2 \times n^2$ matrix. The propagator sends the slowly-decaying modes of $\LL$ (eigenvalues of real part closest to $0$) to the eigenvalues of \emph{largest} modulus of $e^{\LL t}$, so Arnoldi naturally converges to the low-lying spectrum.

This spectral reordering plays exactly the role of the shift-invert transformation we use below, with $e^{\LL t}$ in place of the resolvent $(\mu - \LL)^{-1}$; the two differ in how the transformed operator is applied. Arnoldi--Lindblad integrates the master equation, at a cost governed by the stiffness of the dynamics, whereas shift-invert solves a linear system, which is precisely what the preconditioner of \cref{sec:preconditioner} makes inexpensive.

\subsubsection{New method based on a preconditioned shift-invert}

We propose to apply Arnoldi to the shift-inverted deflated Lindbladian $(\mu - \widetilde\LL_\eta)^{-1}$, never forming the $n^2 \times n^2$ Liouvillian: the deflation removes the zero mode, the shift-invert turns the targeted eigenvalues into the dominant ones, and the preconditioner of \cref{sec:preconditioner} makes the resulting inner linear solves rather inexpensive. We now detail these three ingredients.

\paragraph{Deflating the zero mode.}
The Lindbladian is singular: its kernel is spanned by the steady state, and the vectorized identity is a left null vector. Inverting at the shift $\mu = 0$ would thus invert a singular operator. We remove the zero mode by deflation, using the operator $\widetilde\LL_\eta(\rho) = \LL(\rho) + \eta\,\Tr(\rho)\,\Id$ of \cref{eq:steady_state_linear_problem}. Since any eigenvector of $\LL$ associated with a nonzero eigenvalue $\lambda$ is traceless (taking the trace of $\LL\rho = \lambda\rho$ gives $\lambda\Tr\rho = 0$), this leaves the nonzero eigenpairs of $\LL$ unchanged and replaces the zero eigenvalue by $\eta\, n$. The constant $\eta > 0$ is free and only fixes where the deflated mode lands: it should be taken large enough for $\eta n$ to sit well away from the eigenvalues we target, so that the shift-invert below does not turn it into a dominant one. We take $\eta = 1$ in the benchmarks, and $\eta = 10/n$ in \cref{fig:spectrum_shift_invert} to better illustrate this step. The only excluded value is $\eta = 1/\Tr\RR_\mu^\SS(\Id)$, at which the deflated preconditioner of \cref{eq:ShermanMorr} degenerates.

\begin{figure}[htb]
    \centering
    \includegraphics[width=\linewidth]{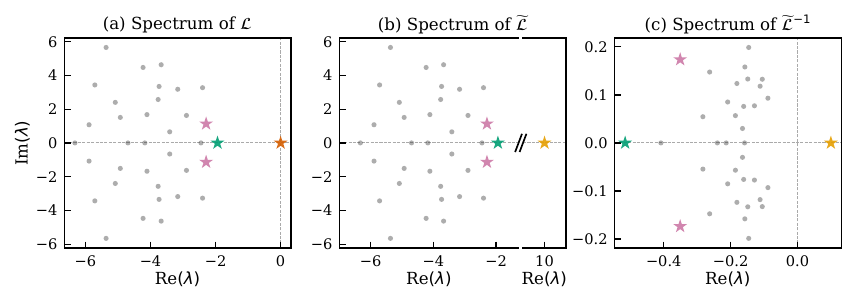}
    \caption{Random dense Lindbladian (\cref{app:dense_gen}) spectra for a Hilbert space of size $n=6$.
        (a)~Full spectrum of $\mathcal{L}$.
        (b)~Spectrum of $\widetilde\LL_\eta = \LL + \eta\,\Tr(\cdot)\,\Id$ with $\eta = 10/n$, which shifts the steady-state eigenvalue to $\eta n = 10$ (orange star) while leaving all nonzero modes unchanged.
        (c)~Spectrum of $\widetilde\LL_\eta^{-1}$: the low-lying modes of $\widetilde\LL_\eta$ (small $|\lambda|$) become its dominant eigenvalues (large $|1/\lambda|$). The operator actually iterated in the text is $(\mu - \widetilde\LL_\eta)^{-1}$, which at $\mu = 0$ differs from $\widetilde\LL_\eta^{-1}$ by an overall sign.
        In all panels, the red star marks the steady state, the green star the slowest-decaying nonzero mode, and the purple stars the next conjugate pair.
        The orange star in the second and third panels~(b--c) marks the deflated steady-state eigenvalue.}
    \label{fig:spectrum_shift_invert}
\end{figure}

\paragraph{Shift-invert Arnoldi.}
The targeted eigenvalues sit near $0$, where Arnoldi applied to $\widetilde\LL_\eta$ directly converges slowly, since Arnoldi resolves the eigenvalues of largest magnitude first. We therefore apply it to the shift-inverted operator $(\mu - \widetilde\LL_\eta)^{-1}$. This operator has the same eigenvectors as $\widetilde\LL_\eta$, and an eigenvalue $\lambda$ of $\widetilde\LL_\eta$ is sent to the eigenvalue $\nu = 1/(\mu - \lambda)$. The eigenvalues $\lambda$ closest to the shift $\mu$ thus become the $\nu$ of largest modulus, which Arnoldi resolves first; each is mapped back through $\lambda = \mu - 1/\nu$. Taking $\mu = 0$ targets the smallest-magnitude nonzero eigenvalues of $\LL$.

\paragraph{The inner linear solves.}
Each Arnoldi step applies the shift-inverted operator once, i.e.\ solves a single linear system in $\mu - \widetilde\LL_\eta$, which we do with the right-preconditioned GMRES of \cref{sec:precond_def}. The preconditioner is the no-jump resolvent $\RR_\mu^\SS$, corrected for the deflation: $\widetilde\LL_\eta$ differs from $\LL$ by the rank-one term $\eta\,\ket{\Id}\!\bra{\Id}$, where $\bra{\Id}\rho\rangle = \Tr\rho$, and the Sherman--Morrison formula \cref{eq:ShermanMorr} folds that term into the preconditioner exactly, at the price of one trace and one extra Lyapunov solve. The latter being independent of the right-hand side, it is computed once and reused throughout the Arnoldi run (\cref{alg:precond}). Successive Arnoldi steps moreover solve nearby systems: rather than restarting each solve from scratch, we recycle the Krylov subspace and warm-start it with the previous solution \cite{Parks2006Oct,Soodhalter2014Jul}. The full method is given in \cref{alg:low-lying-spectrum} of \cref{sec:linear_solver}.

\subsubsection{Numerical benchmarks}

We benchmark on the same two families as in \cref{sec:benchmark_steady_state}: the dense random Lindbladian and the memory--buffer cat qubit described in \cref{sec:cat}. Because our shift-invert Arnoldi never forms the $n^2 \times n^2$ Liouvillian, it computes the low-lying spectrum of systems roughly one order of magnitude larger than the SciPy baselines, which both form the $n^2 \times n^2$ Liouvillian and run out of memory beyond a few hundred basis states (\cref{fig:benchmark-lowlying-spectrum}). This reach lets us resolve the low-lying spectrum of the cat qubit across cat sizes (\cref{fig:spectrum-eigenvalues-cat}), recovering its characteristic noise bias: an exponentially suppressed bit-flip rate together with linearly growing phase-flip rates, the exact feature that makes the cat qubit attractive for quantum error correction.

\begin{figure}[tb]
    \centering
    \includegraphics[width=\linewidth]{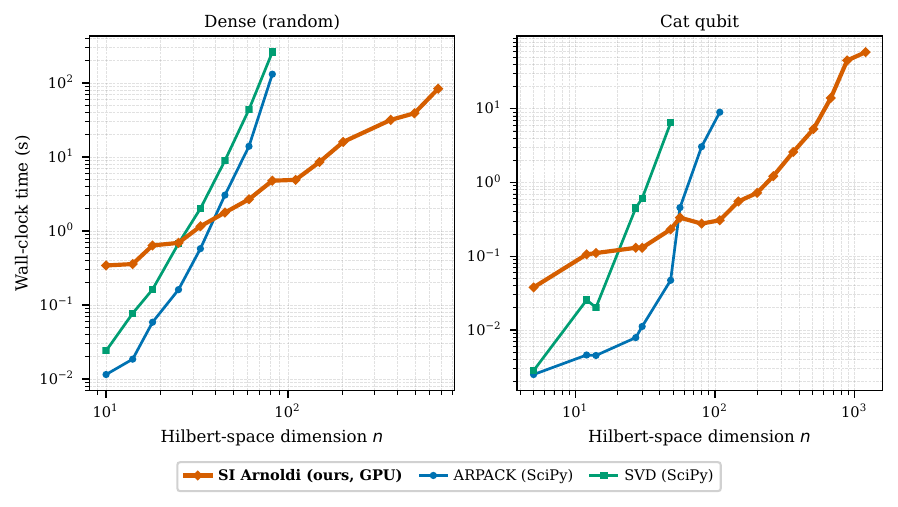}
    \caption{Wall-clock time to compute the low-lying spectrum as a function of the Hilbert-space dimension $n$ (log--log), for the dense random Lindbladian (left, \cref{app:dense_gen}) and the memory--buffer cat qubit (right, \cref{sec:cat}). Our shift-invert Arnoldi solver (SI Arnoldi, ours) is compared against the dense eigenvalues solver of NumPy and the ARPACK shift-invert solver of SciPy, both applied to the $n^2 \times n^2$ Liouvillian. See \cref{sec:details_num_benchmark} for the exact configurations.}
    \label{fig:benchmark-lowlying-spectrum}
\end{figure}

\begin{figure}[tb]
    \centering
    \includegraphics[width=0.6\linewidth]{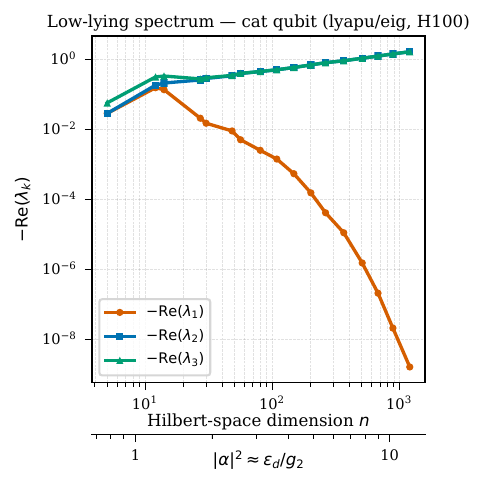}
    \caption{Low-lying spectrum of the memory--buffer cat qubit (\cref{sec:cat}) computed with our solver: the three smallest nonzero decay rates $-\Re(\lambda_k)$ against the Hilbert-space dimension $n$ (lower axis: the corresponding cat size $|\alpha|^2 \approx \varepsilon_d/g_2$). The cat code space is two-dimensional, so under ideal two-photon stabilization the four operators acting on it would all be stationary; the residual single-photon loss lifts this degeneracy into a unique steady state and three slow modes. The slowest, $\lambda_1$, is the bit-flip, whose rate is exponentially suppressed in $|\alpha|^2$; the other two, $\lambda_2$ and $\lambda_3$, are phase-flip modes whose rates grow linearly in $|\alpha|^2$, see \eg \cite{PhysRevA.111.012617}.}
    \label{fig:spectrum-eigenvalues-cat}
\end{figure}

\subsection{Implicit solver}
\label{sec:implicit}

\begin{figure}[tb]
    \centering
    \includegraphics[width=\linewidth]{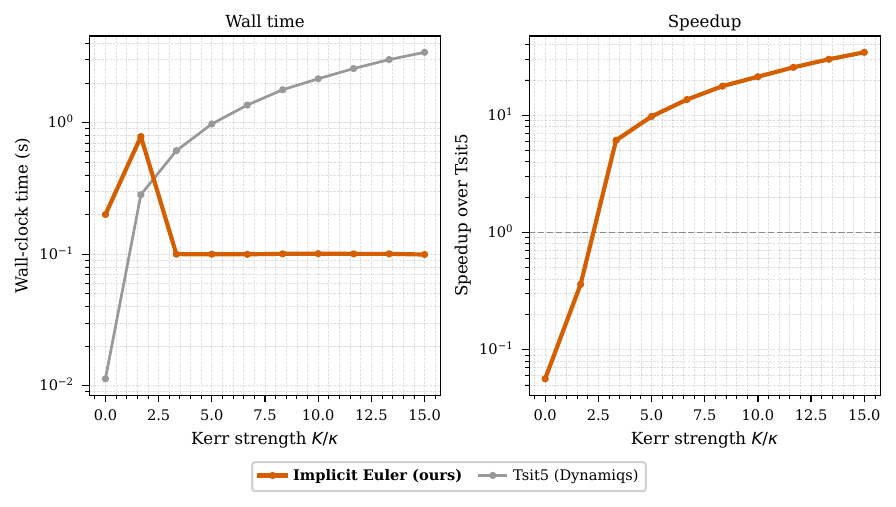}
    \caption{Comparison of our implicit Euler scheme with a generic explicit integrator (Tsit5) from dynamiqs \cite{guilmin2025dynamiqs}. On stiffer systems with high Kerr, the implicit Euler scheme takes fewer steps, resulting in a speedup with respect to the explicit integrator. }
    \label{fig:benchmark-ie}
\end{figure}

Integrating the master equation $\dot{\rho}=\LL\rho$ in high-dimensional systems is often hampered by stiffness: explicit integrators (\eg Runge--Kutta) require a time step $\Delta t$ small enough to resolve the fastest modes, and this restriction typically tightens with the truncation size when simulating truncated approximations of infinite-dimensional, unbounded Lindblad dynamics (see \cite[Section 4]{robin2025unconditionally}). By contrast, implicit schemes remain stable for much larger $\Delta t$ and are therefore better suited to stiff problems.

Implicit methods have seen limited application to Lindblad equations because each time step entails solving an $n^2\times n^2$ linear system. The preconditioner of \cref{sec:preconditioner} makes these linear solves practical even for large, dense systems, enabling efficient implicit time stepping. For the implicit Euler scheme, the step to be inverted reads
\begin{equation}\label{eq:resolvent}
    \rho_{t+\Delta t} = \rho_t + \Delta t \LL \rho_{t+\Delta t} \implies \rho_{t+\Delta t} = (\Id - \Delta t \LL)^{-1} \rho_t = \frac{1}{\Delta t} \RR_{1/\Delta t}^\LL \rho_t.
\end{equation}
This requires computing the action of the resolvent at each time step. We do this by preconditioning the linear system in \cref{eq:resolvent} with $\RR_\lambda^\SS$ at the shift $\lambda = 1/\Delta t$, computed exactly with the same machinery as for the steady state in \cref{sec:steady_state_linear_solve}. This is the one application in which the shift is strictly positive, so that \cref{thm:contraction} applies and \cref{rk:gmres-rate} guarantees a geometric convergence rate for the inner solves.

\paragraph{Benchmarks.}
We consider the random Lindbladian system introduced in \cref{sec:systems_studied}. We add a diagonal quartic term $K\,N^2$, with $N = \mathrm{diag}(0,1,\dots,n-1)$, to its Hamiltonian to have a knob on the stiffness of the resulting ODE. The step size of the implicit Euler scheme is chosen so that all the states of the simulation are close enough to a reference (expensive) simulation $\rho_\star$. That is, we choose $\Delta t$ s.t.
\begin{equation*}
    \mathcal{F}(\rho_\mathrm{IE}(k \Delta t), \rho_\star(k \Delta t)) \geq 0.99,
\end{equation*}
with $\mathcal{F}$ the fidelity between two density matrices.

Results are presented in \cref{fig:benchmark-ie}. We see that the implicit Euler scheme requires a roughly constant number of steps regardless of the Kerr strength. This results in a large speedup over the explicit integrator as soon as the system is stiff enough. Further work would be required to make this a competitive solver in practice. Notably, adaptive step sizing as well as higher-order methods would be key.

\section{Conclusion and perspectives}
\label{sec:conclusion}

The methods of this paper all rest on a single observation: the Lindbladian splits as
$\LL = \SS + \KK$, and the no-jump part $\SS$, however ill-conditioned $\LL$ may be, can be inverted
exactly at cost $O(n^3)$ by solving a continuous Lyapunov equation. From it we built two objects:
the CPTP map $\Phi = -\KK\SS^{-1}$, whose fixed point yields the steady state and whose eigenvalue
$1$ is simple and of maximal modulus (\cref{thm:steady-state-eigenvalue}); and the no-jump resolvent
$\RR_\lambda^\SS$, a right preconditioner for $\lambda - \LL$ whose error operator is a strict
contraction for every $\lambda > 0$ (\cref{thm:contraction}), which yields both an explicit series
for the resolvent of the Lindbladian (\cref{thm:resolvent}) and a geometric convergence rate for the
preconditioned GMRES (\cref{rk:gmres-rate}). Three distinct tasks then reduce to the same building
block: computing the steady state, computing the low-lying spectrum by shift-invert Arnoldi, and
integrating stiff dynamics implicitly. All of them manipulate $n \times n$ matrices only, never
forming the $n^2 \times n^2$ Liouvillian, and reduce to dense matrix products. This is why we measure large
gains, further amplified on GPU.

Several directions remain open.
    \paragraph{Choosing the jump/no-jump split.}
    We have not leveraged the gauge freedom of \cref{rk:gauge} in the present work, and we believe that an optimization on this additional degree of freedom should lead to quantitative improvements. We leave it to future work. Let us nonetheless give two theoretical remarks. First, assuming the number of jump operators $d$ to be at least one, all gauge changes but a finite number make $\tilde\SS$ invertible:
    retracing the proof of \cref{thm:steady-state-eigenvalue}, $\tilde\SS$ is singular only if some
    $\psi \neq 0$ satisfies $L_j\psi = -c_j\psi$ for every $j$. Second, the peripheral eigenvalues of \cref{rk:peripheral} are not intrinsic to $\LL$ but to
    the split: on that example the shift $c = (1,0,0)$ leaves $1$ simple and brings every other eigenvalue of $\tilde\Phi$ to modulus at most $\sqrt{3/8}$, the pair $e^{\pm 2i\pi/3}$ becoming  $\frac12 \pm \frac{i}{2\sqrt2}$.
\paragraph{Beyond first order in time.}
The implicit integrator of \cref{sec:implicit} is a proof of concept. Adaptive step-size selection
and higher-order multistep schemes of backward differentiation formula (BDF) type would be needed to make it practical, and both fit
the present machinery well: a BDF step requires inverting $\Id - \beta\,\Delta t\,\LL$ for a
scheme-dependent constant $\beta$, that is, up to the prefactor $1/(\beta\,\Delta t)$, the same resolvent at the shift
$\lambda = 1/(\beta\,\Delta t)$. Changing the order or the step size only moves that shift, which
enters our Lyapunov backend through an element-wise division (\cref{alg:lyapu-eig}); the
eigendecomposition of $G$ is computed once and reused throughout.

\section*{Acknowledgements}

We warmly thank Pierre Guilmin, Pierre Rouchon, Alain Sarlette, and Hector Hutin for helpful
discussions and comments on this work.

 This project has received funding from the European Research Council (ERC) from the QFT.zip
project (grant agreement No. 101040260) and Plan France 2030 through the project ANR-22-PETQ-0006.

\newpage
\appendix
\section{Numerical implementation}
\label{sec:num_implementation}

All three solvers of the main text rely on the same tool: the exact inversion of the
no-jump part $\SS$, which amounts to solving a continuous Lyapunov equation. We first
describe how we solve it (\cref{sec:continuous_lyapunov}), then how the steady-state and
low-lying-spectrum solvers are built on top of it (\cref{sec:linear_solver}). From this exact inversion, we introduce two algorithms, which the three solvers leverage: the no-jump resolvent $\RR_\lambda^\SS$
(\cref{alg:lyapu-eig}) and the right preconditioner it induces (\cref{alg:precond}).

\subsection{Solving the continuous Lyapunov equation}
\label{sec:continuous_lyapunov}

We detail the standard methods to solve the continuous Lyapunov equation $\SS(X) = GX + XG^\dag = Y$,
and their shifted variant $(\lambda - \SS)(X) = Y$ that yields the resolvent
$\RR_\lambda^\SS = (\lambda - \SS)^{-1}$ used as a preconditioner.

\subsubsection{Bartels--Stewart algorithm}

The standard algorithm is the Bartels--Stewart algorithm \cite{bartels1972algorithm,golub1979hessenberg}.
It computes the Schur decomposition $G = U T U^\dagger$ with $T$ upper triangular, transforming
the equation into $T \tilde X + \tilde X T^\dagger = \tilde Y$ with $\tilde X = U^\dagger X U$,
$\tilde Y = U^\dagger Y U$, which is solved by back-substitution. When $Y$ is Hermitian, so is $X$,
halving the cost. The shift is applied on the triangular factor, $T \to T - \tfrac{\lambda}{2}\Id$,
which realizes $-\RR_\lambda^\SS$. The back-substitution is sequential and hard to parallelize on
GPUs, and can suffer from stability issues; a mixed-precision refinement mitigates the
latter \cite{dmytryshyn2025mixed}.

\subsubsection{Eigendecomposition method}

Alternatively, when $G$ is diagonalizable, we use its eigendecomposition
$G = U \Sigma U^{-1}$, $\Sigma = \mathrm{diag}(\sigma_1, \dots, \sigma_n)$. Writing
$V = U^{-\dagger}$ and $\hat X = V^\dagger X V$, $\hat Y = V^\dagger Y V$, the shifted equation
decouples element-wise,
\begin{equation}
    (\lambda - \sigma_i - \sigma_j^*)\, \hat X_{ij} = \hat Y_{ij}
    \implies
    \hat X_{ij} = \frac{\hat Y_{ij}}{\lambda - \sigma_i - \sigma_j^*},
\end{equation}
so that
\begin{equation}\label{eq:lyapu_eig_solution}
    \RR_\lambda^\SS(Y)
    = U \left[ \left(\frac{1}{\lambda - \sigma_i - \sigma_j^*} \right)_{ij}
        \circ \big(V^\dagger Y V\big) \right] U^\dagger ,
\end{equation}
with $\lambda = 0$ recovering $-\SS^{-1}$. Since $\Re(\sigma_i) < 0$ under the assumption that
$\SS$ is invertible, the denominators never vanish for $\lambda \ge 0$.

In other words, this exhibits a basis in which the superoperator $\SS$ is diagonal. In particular, this provides a fast, matrix-matrix implementation of the inverse.

\subsubsection{Comparison and implementation}

Both methods scale as $O(n^3)$. The Bartels--Stewart algorithm has the smaller leading constant, but once computed, the
eigendecomposition route requires matrix--matrix products (GEMM), which are
extremely optimized on GPUs.
The Bartels--Stewart variant, in contrast, relies on a Schur factorization and a subsequent back substitution. We did not find an efficient implementation of the latter on GPU.
Because our workflow solves the Lyapunov equation repeatedly inside a Krylov solver, we
precompute the eigendecomposition once and reuse it across all solves. The per-solve cost is
four dense $n\times n$ products and one element-wise division.

\begin{algorithm}[htbp]
    \caption{No-jump resolvent $\RR_\lambda^\SS$ via eigendecomposition}
    \label{alg:lyapu-eig}
    \KwIn{generator $G = -iH - \tfrac12\sum_k L_k^\dagger L_k$, shift $\lambda \ge 0$, operator $Y$}
    \KwOut{$X = \RR_\lambda^\SS(Y) = (\lambda - \SS)^{-1}(Y)$}
    \textbf{Precompute (once):} eigendecomposition $G = U\Sigma U^{-1}$; store $U$ and $V \gets U^{-\dagger}$
    \tcp*{$O(n^3)$, general eigensolver}
    \textbf{Apply:} $X \gets U\left[\dfrac{(V^\dagger Y V)_{ij}}{\lambda - \sigma_i - \sigma_j^*}\right]U^\dagger$
    \tcp*{$O(n^3)$, GEMM-bound}
\end{algorithm}

\subsection{Linear solver implementation}
\label{sec:linear_solver}

\paragraph{GMRES.}
Every linear system we solve is a shifted, possibly deflated Lindbladian preconditioned on the
right by the no-jump resolvent of \cref{alg:lyapu-eig}. Our GMRES is a custom, matrix-free
implementation in JAX \cite{jax2018github}: it never forms the $n^2\times n^2$ Liouvillian and applies $\LL$, $\RR_\lambda^\SS$
and the preconditioner as functions on $n\times n$ matrices. It is a right-preconditioned
restarted GMRES with Krylov subspace recycling: after each cycle we retain a small
harmonic-Ritz subspace $(U,C)$ with $\LL\,\RR_\lambda^\SS\,U = C$, which is passed to the next
cycle and across successive solves, and each solve is warm-started from the previous iterate. The
least-squares problem is solved by a QR factorization of the augmented Hessenberg matrix, and the
Arnoldi basis is fully re-orthogonalized. For the
steady state the stopping criterion is the size-independent $\norm{\LL(\hat\rho)}_{\max} < \varepsilon$ on the
hermitized and trace-normalized state $\hat\rho$.

Written in the form $\lambda - \widetilde\LL_\eta = (\lambda - \LL) - \eta \,\ket{\Id}\!\bra{\Id}$ of \cref{sec:precond_def}, the two deflated
systems of \cref{sec:fixed_point,sec:low_lying} carry the rank-one term $-\eta\,\Tr(\cdot)\,\Id = -\eta\,\ket{\Id}\!\bra{\Id}$,
in addition to the no-jump part. We add it into the preconditioner exactly with the
Sherman--Morrison formula, written here for the shift $\lambda = 0$ (recall $\RR_0^\SS = (-\SS)^{-1}$):
    \begin{equation}\label{eq:ShermanMorr}
        \big({-\SS} - \eta\,\ket{\Id}\!\bra{\Id}\big)^{-1}(X)
        = \RR_0^\SS(X)
        + \frac{\eta\,\Tr\!\big(\RR_0^\SS(X)\big)}{1 - \eta\,\Tr\!\big(\RR_0^\SS(\Id)\big)}\,\RR_0^\SS(\Id),
    \end{equation}
    which is defined as soon as $\eta\,\Tr\!\big(\RR_0^\SS(\Id)\big) \neq 1$, the value at which the deflated
    no-jump operator ${-\SS} - \eta\,\ket{\Id}\!\bra{\Id}$ is itself singular. Here $\RR_0^\SS(\Id)$ and its
    trace are precomputed once, so the correction costs only one extra base solve and a trace, and the
    check above comes for free. It excludes a single value of $\eta$, far from the ones we use:
$\Tr \RR_0^\SS(\Id) = \int_0^\infty \Tr\big(e^{tG}e^{tG^\dag}\big)\dd t$ is of the order of the dimension
    times a decay time, so its inverse is small compared to $\eta = 1$. Written for a general shift $\lambda$, i.e. replacing $-\SS$ by $\lambda-\SS$ and $\RR_0^\SS$ by $\RR_\lambda^\SS$, this defines the
    preconditioner \texttt{Precond} of \cref{alg:precond}.

    \begin{algorithm}[htbp]
    \caption{Deflated no-jump right preconditioner (\texttt{Precond})}
    \label{alg:precond}
    \KwIn{operator $X$, shift $\lambda$, deflation coefficient $c$}
    \KwOut{$P_\lambda(X)$, an approximation of the deflated $\big((\lambda-\LL) - c\,\ket{\Id}\!\bra{\Id}\big)^{-1}(X)$}
$Z \gets \RR_\lambda^\SS(X)$ \tcp*{\cref{alg:lyapu-eig}}
$W \gets \RR_\lambda^\SS(\Id)$ \tcp*{precomputed once}
    \Return $Z + \dfrac{c\,\Tr(Z)}{1 - c\,\Tr(W)}\,W$ \tcp*{Sherman--Morrison, \cref{eq:ShermanMorr}}
    \end{algorithm}

    We now assemble the three solvers. The steady state can be obtained either as the fixed point of
$\Phi$ (\cref{alg:steady-state-arnoldi}) or as the solution of the deflated linear system
    (\cref{alg:steady-state}); the low-lying spectrum uses shift-invert Arnoldi
    (\cref{alg:low-lying-spectrum}).

    \begin{algorithm}[htbp]
    \caption{Steady state via preconditioned GMRES}
    \label{alg:steady-state}
    \KwIn{$H$, $\{L_k\}$, deflation $\eta$, tolerance $\varepsilon$, Krylov size $m$}
    \KwOut{$\hat\rho$ with $\|\LL(\hat\rho)\|_{\max} \leq \varepsilon$, $\Tr\hat\rho = 1$}
    Build $\RR_\lambda^\SS$ (\cref{alg:lyapu-eig}) and $P_0 \equiv \texttt{Precond}(\cdot,\lambda{=}0,c{=}\eta)$ (\cref{alg:precond})\;
    Solve $\widetilde\LL_\eta\,\rho = \eta\,\Id$ (\cref{eq:steady_state_linear_problem}) by right-preconditioned GMRES($m$):\;
    \quad recycle the Krylov subspace across restarts; warm-start from $\ket{0}\!\bra{0}$\;
    \quad each cycle: $\hat\rho \gets \tfrac12(\rho + \rho^\dagger)$, $\hat\rho \gets \hat\rho / \Tr\hat\rho$;\ stop if $\norm{\LL(\hat\rho)}_{\max} < \varepsilon$\;
    \Return $\hat\rho$
    \end{algorithm}

    \paragraph{Arnoldi.}
    The steady state also arises as the fixed point of the CPTP map $\Phi = -\KK\SS^{-1}$
    (\cref{thm:steady-state-eigenvalue}), whose dominant eigenvalue is $1$. We find it by an
    explicitly-restarted Arnoldi iteration on $\Phi$ (\cref{alg:steady-state-arnoldi}): each cycle
    builds a Krylov subspace of $\Phi$, keeps the Ritz vector whose Ritz value is closest to $1$, and
    restarts from it, tracking the running-best iterate. As for the linear solver, convergence is
    tested in batches on the reconstructed state.

    \begin{algorithm}[htbp]
    \caption{Steady state as the fixed point of $\Phi$}
    \label{alg:steady-state-arnoldi}
    \KwIn{$H$, $\{L_k\}$, tolerance $\varepsilon$, Krylov size $m$, max restarts $N$}
    \KwOut{$\hat\rho$ with $\LL(\hat\rho) \approx 0$, $\Tr\hat\rho = 1$}
    Build $\SS^{-1} = -\RR_0^\SS$ (\cref{alg:lyapu-eig}); define $\Phi(\xi) = -\KK(\SS^{-1}\xi)$\;
    \For{$j = 1, \dots, N$}{
    Run one Arnoldi($m$) cycle on $\Phi$; let $\xi$ be the Ritz vector with Ritz value nearest $1$\;
$\rho \gets -\SS^{-1}(\xi)$;\ $\hat\rho \gets \tfrac12(\rho+\rho^\dagger)$;\ $\hat\rho \gets \hat\rho/\Tr\hat\rho$\;
    \If{$\norm{\LL(\hat\rho)}_{\max} < \varepsilon$}{\Return $\hat\rho$}
    restart from $\xi$\;
    }
    \Return running-best $\hat\rho$
    \end{algorithm}

    The low-lying spectrum is obtained by shift-invert Arnoldi on the deflated Lindbladian
$\widetilde\LL_{\eta}$ (\cref{alg:low-lying-spectrum}). We iterate on
$(\mu - \widetilde\LL_{\eta})^{-1}$, which maps the eigenvalues of $\LL$ nearest $\mu$ to
    the dominant eigenvalues $\nu$, recovered by $\lambda = \mu - 1/\nu$; taking $\mu = 0$
    targets the smallest-magnitude nonzero modes. Each Arnoldi step applies this operator once, i.e.\
    solves one linear system with the preconditioned recycled GMRES above. We use a thick-restart
    (Krylov--Schur) scheme, keeping the best Ritz pairs between cycles, and thread the inner GMRES
    recycling subspace through every matvec. The residual is measured directly on the true
    Lindbladian, $\max_k \norm{\LL v_k - \lambda_k v_k}_{\max} < \varepsilon$, on the
    raw eigenvectors (which are traceless and are \emph{not} hermitized or normalized).

    We also tried the default ARPACK algorithm of Scipy (\texttt{scipy.sparse.linalg.eigs}), but as it does not benefit from end-to-end JAX compilation, it is slower by a constant factor.

    \begin{algorithm}[htbp]
    \caption{Low-lying spectrum via shift-invert Arnoldi}
    \label{alg:low-lying-spectrum}
    \KwIn{$H$, $\{L_k\}$, shift $\mu$, deflation $\eta$, number of modes $k$, tolerance $\varepsilon$, Krylov size $m$, Ritz pairs to keep $p$}
    \KwOut{$k$ smallest-magnitude nonzero eigenpairs $(\lambda_k, v_k)$ of $\LL$}
    Build $\RR_\mu^\SS$ (\cref{alg:lyapu-eig}) and $P_\mu \equiv \texttt{Precond}(\cdot,\mu,\eta)$ (\cref{alg:precond})\;
    Define $\mathrm{OP}(b)$: solve $(\mu - \widetilde\LL_{\eta})\,x = b$ by right-preconditioned recycled GMRES\;
    \Repeat{$\max_k \norm{\LL v_k - \lambda_k v_k}_{\max} < \varepsilon$ \emph{or} max restarts}{
    Extend the Arnoldi factorization of $\mathrm{OP}$ to size $m$ (recycling threaded through each matvec)\;
    Ritz values $\nu$ of $\mathrm{OP}$ $\to$ eigenvalues $\lambda = \mu - 1/\nu$ of $\LL$\;
    Thick-restart: keep the $p$ best Ritz pairs\;
    }
    \Return the $k$ eigenpairs of smallest $|\lambda|$
    \end{algorithm}

    \section{Benchmarks}
    \label{sec:details_num_benchmark}

    \subsection{Description of the system}

    \paragraph{Cat qubit (sparse).}
    \label{sec:cat}
    The sparse benchmark models cat-state stabilization in a memory--buffer architecture. With
$\{a,a^\dagger\}$ and $\{b,b^\dagger\}$ the annihilation/creation operators of the memory and
    buffer modes, the dynamics is set by the Hamiltonian
    \begin{equation}
        H = g_2\left(a^2 b^\dagger + a^{\dagger 2} b\right) - \varepsilon_d \left(b + b^\dagger\right),
    \end{equation}
    and the four jump operators
    \begin{equation}
        \begin{aligned}
            L_1 & = \sqrt{\kappa_b(1+\bar n_b)}\, b, & \quad L_2 & = \sqrt{\kappa_b\, \bar n_b}\, b^\dagger,  \\
            L_3 & = \sqrt{\kappa_a(1+\bar n_a)}\, a, & \quad L_4 & = \sqrt{\kappa_a\, \bar n_a}\, a^\dagger ,
        \end{aligned}
    \end{equation}
    where $\bar n_a$ and $\bar n_b$ denote the thermal occupations of the two modes. This captures
    two-photon exchange between the memory and the lossy buffer together with thermal
    excitation and relaxation in both modes (no Kerr term).

    The resulting Lindbladian is sparse, has a
    small spectral gap that shrinks with cat size, and admits a unique steady state, so $\SS$ is
    invertible. The physical parameters are taken from \cite{QuantumControlCatQubits} and listed in
    \cref{tab:memory_buffer_params}.

    \begin{table}[h]
    \centering
    \caption{Parameters of the memory--buffer Lindbladian used in the benchmarks.}
    \label{tab:memory_buffer_params}
    \begin{tabular}{lc}
    \hline
    Parameter           & Value                \\
    \hline
$g_2/2\pi$          & $0.763~\mathrm{MHz}$ \\
$\kappa_b/2\pi$     & $2.6~\mathrm{MHz}$   \\
$\kappa_a/2\pi$     & $9.3~\mathrm{kHz}$   \\
$\bar n_a$          & $10\%$               \\
$\bar n_b$          & $1.1\%$              \\
    \hline
    \end{tabular}
    \end{table}

The creation and annihilation operators \emph{should} verify $[a,a^\dagger]=1$, $[b,b^\dagger] = 1$, and $[a,b]=[a,b^\dagger]=0$, an algebra that can only be represented in infinite dimension. In practice, for numerical simulations, one truncates $a$ to $n_a$ levels and $b$ to $n_b$ levels. Namely, one takes $a := \tilde{a}\otimes \Id_{n_b}$ and $b := \Id_{n_a} \otimes \tilde{b}$
\[
\tilde{a} \;=\;
\begin{pmatrix}
0 & \sqrt{1} & & & \\
 & 0 & \sqrt{2} & & \\
 & & \ddots & \ddots & \\
 & & & 0 & \sqrt{n_a-1} \\
 & & & & 0
\end{pmatrix}
\in \mathbb{C}^{n_a \times n_a},
~~
\tilde{b} \;=\;
\begin{pmatrix}
0 & \sqrt{1} & & & \\
 & 0 & \sqrt{2} & & \\
 & & \ddots & \ddots & \\
 & & & 0 & \sqrt{n_b-1} \\
 & & & & 0
\end{pmatrix}
\in \mathbb{C}^{n_b \times n_b}.
\]
The effective Hilbert-space dimension is $n = n_a\, n_b$, with $n_a$ and $n_b$ the memory and
    buffer truncations. As we scale $n$, we set the buffer drive $\varepsilon_d = g_2 n_a/5$ so that
    the target cat amplitude $\alpha^2 = \varepsilon_d/g_2 = n_a/5$ stays well within the truncated
    memory space, and choose $n_a = \mathrm{round}(\sqrt{3n})$, $n_b = \lfloor n_a/3\rfloor$. This
    keeps the physically populated Fock states resolved while avoiding truncation artifacts.

    \paragraph{Random (dense).}
    \label{app:dense_gen}
    The dense benchmark uses generic, unstructured generators. We draw matrices
$A_0,\dots,A_3 \in \CC^{n\times n}$ whose real and imaginary parts are i.i.d.\ standard normal,
    and set the Hermitian Hamiltonian and three jump operators
    \begin{equation}
        H = \frac{A_0 + A_0^\dagger}{2}, \qquad L_j = \sqrt{\gamma}\, A_j \quad (j=1,2,3), \qquad \gamma = 0.1 .
    \end{equation}
    A single random stream is drawn in increasing $n$ so the systems are reproducible. All
    generators are dense.

    \subsection{Description of the competing methods}

    We compare against established solvers for each of the three tasks; all run in double
    precision, with convergence declared when $\norm{\LL(\rho)}_{\max} < 10^{-8}$ on the hermitized,
    trace-one state, and each run capped at $100\,\mathrm{s}$. Our own methods run on an Nvidia H100 GPU; the CPU competitors run on an Intel Xeon Platinum 8481C.

    \paragraph{Steady state (\cref{fig:benchmark-steadystate}).}

    \emph{Dense SVD} --- with QuTiP \texttt{steadystate(method='svd')}: forms the dense
$n^2\times n^2$ Liouvillian and solves via SVD. It's exact up to machine precision and complexity scales as $O(n^6)$. Runs on CPU.

    \emph{GMRES\,+\,ILU} --- QuTiP \texttt{method='direct', solver='gmres', use\_precond=True}
    with $\texttt{atol}=\texttt{rtol}=10^{-10}$. This uses SciPy's GMRES with an incomplete-LU left-preconditioner at
    its default fill/drop settings. This is an iterative running on CPU.

    \emph{MUMPS} --- multifrontal sparse $LU$ \cite{MUMPS:1,MUMPS:2} on the singular Liouvillian made
    regular by trace-pinning (row $0$ replaced by the trace constraint, right-hand side a weighted
$e_0$), using column-major vectorization. We use the default MUMPS options. This is a direct, CPU-only method.

    \emph{GMRES (Krylov.jl)} --- Julia \texttt{LinearSolve.jl}/\texttt{KrylovJL\_GMRES} with
$\texttt{reltol}=10^{-8}$ on the same trace-pinned sparse system. Again, an unpreconditioned, iterative, CPU based method.

    \emph{Time evolution} --- Dynamiqs \texttt{mesolve} with the explicit adaptive Tsit5 integrator
    (diffrax), $\texttt{rtol}=\texttt{atol}=10^{-12}$, integrated from $\ket{0}\!\bra{0}$ to a fixed
    horizon ($50\,\mu\mathrm{s}$ dense, $2000\,\mu\mathrm{s}$ cat) after which the residual is checked; runs on GPU.

    Our two entries are the fixed-point Arnoldi on $\Phi$ (\cref{alg:steady-state-arnoldi}) and the
    preconditioned GMRES linear solver (\cref{alg:steady-state}), both with the eigendecomposition
    Lyapunov backend.

    \paragraph{Low-lying spectrum (\cref{fig:benchmark-lowlying-spectrum}).}

    \emph{NumPy dense} --- \texttt{numpy.linalg.eig} of the dense $n^2\times n^2$ Liouvillian, keeping
    the four eigenvalues of smallest magnitude (which include the steady state); exact, $O(n^6)$.

    \emph{ARPACK} --- \texttt{scipy.sparse.linalg.eigs} in shift-invert mode with
$\texttt{k}=4$, $\texttt{sigma}=-10^{-6}$, $\texttt{which='LM'}$, $\texttt{tol}=10^{-6}$; the tiny
    nonzero shift keeps the $LU$ factorization off the singularity of $\LL$; no deflation, no
    preconditioner. Both baselines return four eigenvalues including the steady state, whereas our
    shift-invert Arnoldi (\cref{alg:low-lying-spectrum}) returns the three smallest nonzero modes; convergence is the eigenpair residual to fall below $10^{-4}$ in max-norm.

    \paragraph{Time evolution (\cref{fig:benchmark-ie}).}
    The competitor is again Dynamiqs with the default Tsit5 integrator, here on GPU, with $\texttt{rtol}=\texttt{atol}$ set to the
    loosest of $\{10^{-4},10^{-6},10^{-8},10^{-10}\}$ that meets the accuracy target. Our implicit
    Euler scheme applies the resolvent $(\Id - \Delta t\,\LL)^{-1}$ (\cref{eq:resolvent}) with the Lyapunov preconditioner; in this regime a single preconditioned GMRES cycle per
    step suffices.

    The base system is the dense random Lindbladian ($\gamma = 0.1$) to which we add a diagonal quartic
    term $K\,N^2$, with $N = \mathrm{diag}(0,1,\dots,n-1)$, to the Hamiltonian as a knob on stiffness.
    We fix $n = 32$ and sweep $K \in [0, 15]$. The step size $\Delta t$ is chosen by geometric doubling
    as the largest step whose trajectory stays within fidelity $\mathcal F \ge 0.99$ of a reference
    Tsit5 run at $\texttt{rtol}=\texttt{atol}=10^{-12}$, evaluated at every step time. The implicit
    scheme uses a number of steps essentially independent of $K$, whereas the explicit integrator's
    step count grows with stiffness. Consequently, implicit Euler wins only once the system is stiff
    enough: Tsit5 is faster below the crossover ($K \approx 2.5$), while implicit Euler reaches a
$\sim\!34\times$ wall-clock speedup at $K = 15$. The advantage is specific to this low-precision
    regime; at tighter tolerances the required step size makes implicit Euler uncompetitive, which is
    why we present it as a proof of concept rather than a state-of-the-art solver.

    \paragraph{Resolvent (\cref{fig:benchmark-resolvent}).}
    Each method solves the vectorized shifted system $(\lambda - \LL)(\rho) = \Id/n$, at $\lambda = \lambda_{\max}\cdot 10^{\{0,-1,-2\}}$ with $\lambda_{\max}$ the
    largest-magnitude eigenvalue of $\LL$, estimated once off the clock. As the spectrum of $\LL$
    lies in the closed left half-plane, $\lambda\Id - \LL$ is nonsingular for $\lambda > 0$ and no
    trace-pinning is needed. Convergence is $\norm{(\lambda-\LL)(\rho) - \Id/n}_{\max} < 10^{-8}$ on
    the hermitized state rescaled to its exact trace $\Tr(\rho) = \Tr(b)/\lambda = 1/\lambda$.

    \emph{Dense $LU$} --- \texttt{numpy.linalg.solve} on the dense $n^2\times n^2$ matrix
$\lambda\Id - \LL$; direct, $O(n^6)$, CPU.

    \emph{MUMPS} --- multifrontal sparse $LU$ \cite{MUMPS:1,MUMPS:2} on the sparse $\lambda\Id - \LL$,
    default options; direct, CPU-only.

    \emph{GMRES} --- \texttt{scipy.sparse.linalg.gmres} on $\lambda\Id - \LL$, no preconditioner,
$\texttt{rtol}=0$, $\texttt{atol}=10^{-10}$, restarted every $100$ iterations, at most $5000$
    iterations; iterative, CPU.

    \emph{GMRES\,+\,ILU} --- the same GMRES right-preconditioned by an incomplete-$LU$ factorization
    (\texttt{scipy.sparse.linalg.spilu}) of $\lambda\Id - \LL$ with $\texttt{drop\_tol}=10^{-4}$ and
$\texttt{fill\_factor}=20$; the factorization (computed once) is part of the timed region.

    Our entry is the right-preconditioned recycled GMRES of \cref{sec:precond_def} (Krylov size $32$,
$8$ recycled vectors, $\texttt{atol}=10^{-9}$) with the eigendecomposition Lyapunov backend for
$\RR_\lambda^\SS$.

    \section{\texorpdfstring{$\Phi$}{Phi} is trace-preserving}
    \label{app:properties_phi}
    \begin{lemma}
    Assume $\vecspec(G)\cap i\mathbb{R} = \emptyset$. Then $\Phi=-\KK\SS^{-1}$ is trace-preserving.
    \end{lemma}
    \begin{proof}
    We recall that a CP map is trace-preserving if and only if its dual map is unital, i.e. maps the identity to itself. The dual of $\Phi$ is characterized by the following relation, valid for all operators $X,Y$:
    \begin{align}
        \Tr{\Phi^*(X)Y} & = \Tr{X\Phi(Y)}                                                                  \\
                        & = \Tr{\left(X\int_0^\infty \sum_j L_j e^{tG}Ye^{tG^\dag}L_j^\dag\,dt \right)}    \\
                        & = \Tr{\left(\int_0^\infty \sum_j e^{tG^\dag} L_j^\dag X L_j e^{tG}\,dt Y\right)}
    \end{align}
    So that
    \begin{align}
        \Phi^*(\Id) & = \int_0^\infty e^{tG^\dag}\Big(\sum_i L_i^\dag L_i\Big)e^{tG}\,dt \\
                    & = \int_0^\infty -e^{tG^\dag}(G+G^\dag)e^{tG}\,dt                   \\
                    & = \Big[-e^{tG^\dag}e^{tG}\Big]_0^\infty = \Id,
    \end{align}
    where the boundary term at infinity vanishes because $\Re(\vecspec(G))<0$. Thus the dual is unital and $-\KK\SS^{-1}$ is CPTP.
\end{proof}

\bibliographystyle{plain}
\bibliography{mybib}

\end{document}